\documentclass[letterpaper,twocolumn,10pt]{article}
\usepackage{usenix}
\usepackage{url}
\usepackage{amsmath,amsthm}
\usepackage{epsfig}
\usepackage{graphicx}
\usepackage{authblk}
\usepackage{mhchem}
\usepackage{cite}
\usepackage{color}
\usepackage{hyperref}
\usepackage{xcolor}
\usepackage{array}
\usepackage{hhline}
\usepackage{pifont}
\usepackage[linesnumbered,ruled,vlined]{algorithm2e}

\usepackage{mathtools}
\usepackage[utf8]{inputenc}
\usepackage{amssymb}
\usepackage{enumitem}
\usepackage{bbm,stmaryrd}
\usepackage{mathrsfs}
\usepackage{mathcomp}
\usepackage{multirow}
\usepackage{enumerate}
\usepackage{threeparttable}
\usepackage{makecell}
\usepackage[font=small,skip=0.2pt]{caption}
\usepackage[normalem]{ulem}
\usepackage{scalerel}
\usepackage{tikz}
\usepackage{float}
\usepackage{fontawesome}
\usepackage{color}
\usepackage{listings}
\usepackage{svg}

\usepackage[caption=false, font=footnotesize]{subfig}

\newtheorem{thm}{Theorem}

\theoremstyle{definition}

\newcolumntype{C}[1]{>{\PreserveBackslash\centering}p{#1}}
\newcolumntype{R}[1]{>{\PreserveBackslash\raggedleft}p{#1}}
\newcolumntype{L}[1]{>{\PreserveBackslash\raggedright}p{#1}}

\theoremstyle{remark}

\makeatletter
\def\hlinew#1{\noalign{\ifnum0=`}\fi\hrule \@height #1
\futurelet\reserved@a\@xhline}
\makeatother

\definecolor{greyf}{rgb}{0.7, 0.7, 0.7}
\definecolor{greys}{rgb}{0.85, 0.85, 0.85}
\definecolor{aogreen}{rgb}{0.0, 0.5, 0.0}
\definecolor{carnelian}{rgb}{0.7, 0.11, 0.11}

\newcolumntype{a}{>{\columncolor{Gray}}c}
\newcolumntype{b}{>{\columncolor{white}}c}

\def\revise#1{{\color{black}#1}}

\newcommand{\truep}{\textcolor{aogreen}{\faCheckSquareO}}
\newcommand{\falsep}{\textcolor{carnelian}{\faTimesCircleO}}
\newcommand{\falsen}{\textcolor{carnelian}{\faCircleO}}

\newcommand{\PreserveBackslash}[1]{\let\temp=\\#1\let\\=\temp}
\newcolumntype{C}[1]{>{\PreserveBackslash\centering}p{#1}}
\newcolumntype{R}[1]{>{\PreserveBackslash\raggedleft}p{#1}}
\newcolumntype{L}[1]{>{\PreserveBackslash\raggedright}p{#1}}

\usepackage{algorithmicx}
\usepackage{algpseudocode}
\usepackage{amssymb} 
\usepackage{ulem}

\begin{document}

\date{}

\title{TAPFixer: Automatic Detection and Repair of Home Automation Vulnerabilities based on Negated-property Reasoning\thanks{This paper has been accepted in USENIX Security 2024. This is its extended version.}}
\author{
{\rm
  Yinbo Yu$^{1,2}$,
  Yuanqi Xu$^{1}$,
  Kepu Huang$^{1}$,
  Jiajia Liu$^{1}$
} \\
{
  $^{1}$\textit{National Engineering Laboratory for Integrated Aero-Space-Ground-Ocean Big Data Application Technology, School of Cybersecurity, Northwestern Polytechnical University, China}
}\\
{
  $^{2}$\textit{Research \& Development Institute of Northwestern Polytechnical University in Shenzhen, China}
}\\
{
\textit{Email: \{yinboyu, liujiajia\}@nwpu.edu.cn, \{yuanqixu,huangkepu\}@mail.nwpu.edu.cn}
}
} 
\maketitle

\begin{abstract}
Trigger-Action Programming (TAP) is a popular end-user programming framework in the home automation (HA) system, which eases users to customize home automation and control devices as expected. However, its simplified syntax also introduces new safety threats to HA systems through vulnerable \revise{rule} interactions. Accurately fixing these vulnerabilities by logically and physically eliminating their root causes is essential before rules are deployed. However, it has not been well studied. In this paper, we present TAPFixer, a novel framework to automatically detect and repair rule interaction vulnerabilities in HA systems. It extracts TAP rules from HA profiles, translates them into an automaton model with physical and latency features, and performs model checking with various \revise{correctness properties. It then uses a} novel negated-property reasoning algorithm to automatically infer a patch via model abstraction and refinement and model checking based on negated-properties. We evaluate TAPFixer on market HA apps (1177 TAP rules and 53 properties) and find that it can achieve an 86.65\% success rate in repairing rule interaction vulnerabilities. We additionally recruit 23 HA users to conduct a user study that demonstrates the usefulness of TAPFixer for vulnerability repair in practical HA scenarios.

\end{abstract}

\vspace{-2mm}
\section{Introduction}

A recent spurt of progress in advanced technology (e.g., artificial intelligence, 5G, and cloud computing) has incredibly improved the automation and intelligence of the Internet of Things (IoT). To better provide automatic service for end users, an IoT programming framework named \textit{Trigger-Action Programming} (TAP) has been widely applied in many home automation (HA) platforms, e.g., Samsung SmartThings \cite{smartthings}, Apple Homekit \cite{homekit}, IFTTT \cite{ifttt}, Home Assistant \cite{homeassistant}, Mi Home\cite{mi}, and so on. In general, a TAP rule is defined in the form of ``\textbf{IF} the \textit{trigger} occurs \textbf{WHILE} the \textit{condition} is met, \textbf{THEN} perform the \textit{action}.'' It describes how to operate a device (\textit{action}) under the state constraint (\textit{condition}) when an event or a state change (\textit{trigger}) occurs in the HA system. For example, ``\textbf{IF} the user is present \textbf{WHILE} the \ce{CO2} is above a predefined value, \textbf{THEN} open the window for 10 min''. TAP rules have greatly facilitated and enriched users' lives.

TAP rules make it easier to connect various devices for collaborative automation. However, the surge in interactions between rules, especially those involving the complicated physical space (\textit{i.e.}, the HA physical environment), challenges the correctness of TAP-based HA systems \cite{birnbach2019peeves,chi2020cross,ding2021iotsafe}.
For example, turning the heater on can increase the temperature, which can then activate or enable other TAP rules.
End users have little understanding of the incomprehensible relationship between the logical (\textit{i.e.}, smart device behaviors defined by programs) and physical space. Hence, it is hard for users to configure TAP rules that align with their intentions \cite{zhang2019autotap}, as well as manually identify and fix rule interaction vulnerabilities, resulting in \revise{unexpected device status (e.g., the AC and window are both turn on) or even safety risks (e.g., the door is unlocked when the user is not at home)} in the HA system.

Recently, a significant number of advanced techniques have been proposed to secure TAP-based HA systems, including automation generation \cite{manandhar2020towards, zhang2019autotap}, threat detection \cite{alhanahnah2020scalable, yu2022tapinspector}, access control \cite{celik2019iotguard, schuster2018situational}, privacy leakage analysis\cite{bastys2018if, hsu2019safechain}, and anomaly detection\cite{Sikder_Aksu_Uluagac_2017, birnbach2019peeves, Fu_Zeng_Du_2021}.
While these techniques make great contributions to analyzing rule interaction vulnerabilities, there is a noticeable lack of attention to resolving and preventing them.
Some dynamic control-based methods \cite{jia2017contexlot, tian2017smartauth,ding2021iotsafe,Chi_Zeng_Du} are proposed to control rule enforcement at runtime to avoid risks according to specified safety policies. However, they do not eliminate the root cause of vulnerabilities \revise{(\textit{i.e.}, rule semantic flaws)} and can introduce additional running overhead.
In contrast, static-based methods \cite{liang2016systematically,bu2018systematically, zhang2019autotap} can generate rule patches to \revise{correct rules}. Unfortunately, these methods neglect dynamic factors (e.g., latency and physical interactions) that may change the practical effect of rule executions, and therefore have limited repair capabilities \cite{birnbach2019peeves, chi2020cross}.

To address these limitations, in this paper, we focus on vulnerability static repairing and design TAPFixer, an automatic vulnerability detection and repair framework for securing TAP-based home automation.
To our best knowledge, TAPFixer is the first work that can essentially detect and fix rule interaction vulnerabilities both in the logical and physical space.
It is orthogonal to dynamic control-based methods in that it can statically reduce risks as possible before rule execution, allowing the latter to run with a lower running cost (\textit{i.e.,} few enforcement policies) to prevent risks caused by unpredictable events (e.g., human interference with devices).

Our major contributions are summarized as follows:

\begin{itemize}[leftmargin=2mm, itemindent=1mm, itemsep=2pt,topsep=0pt,parsep=0pt]
\item We design a formal model of rule interactions using finite automaton, which formalizes and embeds physical operating features into rule syntax to enable accurate static vulnerability detection. With such a model and a set of designed correctness properties, TAPFixer detects rule interaction vulnerabilities through model checking.
\item We design a novel negated-property reasoning algorithm that can automatically construct rule patches to fix and radically eliminate vulnerabilities both in the logical and physical space. Given a violation of a property (\textit{i.e.}, a counterexample), its core idea is to use negated properties to reason about negated counterexamples through an abstraction and refinement process, thereby identifying possible repair patches of the violation.

\item We conduct a benchmark to evaluate the accuracy of TAPFixer in comparison to existing approaches. We then apply TAPFixer to market HA apps, where TAPFixer can obtain an 86.65\% success rate in repairing found vulnerabilities on average. We also conduct a user study and performance analysis of TAPFixer to demonstrate how well it detects and repairs rule interaction vulnerabilities.
\end{itemize}

\vspace{-2mm}
\section{Preliminary}
\subsection{TAP Rule Formulation}
In this section, we formulate TAP rules and their interactions based on existing advanced research \cite{chi2022delay,ozmen2022discovering}.
The execution of a TAP rule follows the logical sequence of trigger-condition-action: $r:=\langle t,c\rangle\mapsto a$, where $t$ is a \textit{trigger}, $c$ is a \textit{condition}, $a$ is an \textit{action}, and $\mapsto$ denotes a relationship of the sequential execution. $r$ can be formulated as a set of constraints and assignments on \textbf{entity attributes} $\mathbbm{A}$ which are the intuitive abstract state expression of entities including logical states (e.g., time) and physical objects (including smart devices and physical attributes, e.g., humidity). To achieve a more accurate formal analysis of rule executions, we classify entity attributes as \textbf{\emph{Immediate}} and \textbf{\emph{Tardy}} types \cite{yu2022tapinspector}. While the former (e.g., the state of a switcher) changes instantaneously, the latter (e.g., temperature) takes a period of time to make changes.

\textbf{Rule Syntax}: for $r:=\langle t,c\rangle\mapsto a$, it can be defined as (1) $t$ specifies a constraint on a certain attribute to activate $r$, e.g., humidity $\textless$ 30$\%$; 2) $c$ is a set of constraints on one or more attributes, and all of them must be evaluated to be true for action executions; 3) $a$ represents one or more assignments on attributes. Actions can be divided into two types: \textbf{\emph{Immediate}} and \textbf{\emph{Extended}}. While the former (e.g., turning lights on) completes instantaneously, the latter (e.g., lowering the temperature to 20\textcelsius) needs a period of time to complete.

\textbf{Rule Semantic}: it is the abstract template describing characteristics of entities in TAP rules (e.g., working modes of the security camera), which can be formalized as follows:
\vspace{-2mm}
\begin{equation}
r_S:= t_S\times c_S \times a_S,
\end{equation}
where $t_S$, $c_S$, and $a_S$ denote the semantics of trigger, condition, and action, respectively.

\textbf{Rule Configuration}: it is the abstract representation of attribute values used in TAP rules (e.g., night mode of the security camera), which can be formalized as follows:
\vspace{-2mm}
\begin{equation}
r_C:= t_C\times c_C \times a_C,
\end{equation}
where $t_C$, $c_C$, and $a_C$ denote the configuration of trigger, condition, and action, respectively.

\vspace{-2mm}
\subsection{Rule Interaction Vulnerabilities}\label{vulner}
\revise{Using simplified rule syntax to automate complex HA environments is hard to avoid without errors, even for professional users \cite{zhang2019autotap}.
The rule vulnerabilities may be caused by several reasons \cite{birnbach2019peeves, chi2020cross, wang2019charting}, e.g., user misconfiguration, misleading or deceptive rule apps, or state deception via device spoofing or channel injection. The fundamental root cause is logical flaws in rule semantics. These flaws are compounded during rule interactions among rules through triggers, conditions, or actions in logical or physical channels, which can lead to more unexpected or even vulnerable states.
Existing works \cite{chi2020cross,yu2022tapinspector,chi2022delay} have summarized various rule interaction vulnerabilities. In this paper, depending on how rules interfere, we categorize existing vulnerabilities into 3 basic patterns and further define 5 expanded vulnerability patterns based on specific interference contexts in each basic pattern. We give a concise description of these 8 patterns in follows and their details are described in Appendix \ref{asec:rivp} due to limited space.}

\textbf{Basic Vulnerability Pattern (V1-V3)}. In general, based on how triggers, conditions, and actions of a rule interfere with others, the basic vulnerability pattern is classified into three categories: Trigger-Interference \textbf{V1}, Condition-Interference \textbf{V2}, and Action-Interference \textbf{V3} vulnerabilities\cite{chi2020cross}. In \textbf{V1}, an action may produce an event that unexpectedly interferes with triggers of other rules and changes the rule context defectively, e.g., opening the light triggers the rule (``if the brightness is high, turn off the light''), leading to light strobe.
In \textbf{V2}, an action may change the condition satisfaction of other rules and put the rule execution at risk, e.g., the user comes home late and interferes with the safety rule conditioned as sleep mode.
In \textbf{V3}, rules with the same or different trigger may send conflicting commands to the same device, e.g., the door opens the moment it locks, which may cause a break-in.

\textbf{Expanded Vulnerability Pattern (V4-V8)}.
With in-depth research on vulnerability detection \cite{chi2022delay,yu2022tapinspector}, the basic vulnerability pattern is expanded by new vulnerabilities in more specific interference contexts. Expanded contexts mainly depict physical-related and latency-related features. The expanded pattern can be classified into 5 categories. In Expanded Trigger-Interference vulnerabilities \textbf{(V4)}, an action may also interfere with triggers of other rules through channel-based interaction describing physical channel shared by actuators and sensors (e.g., both the thermostat and the temperature sensor work on temperature).  \revise{Fig. \ref{fig:rin}(a) shows an example of V4.}
Expanded Action-Interference vulnerabilities \textbf{(V5-V7)} \cite{yu2022tapinspector} includes \emph{Disordered action
scheduling} \textbf{(V5)}, \emph{Action overriding} \textbf{(V6)}, and \emph{Action breaking} \textbf{(V7)}, which describes that an action may override or interrupt actions of other rules due to different time delays.  \revise{Fig. \ref{fig:rin}(b) shows an example of V7.} Expanded Condition-Interference vulnerabilities \textbf{(V8)} mainly occur when rules work on the same channel attribute but have different preferences.

\vspace{-2mm}
\section{Overview}

\vspace{-2mm}
\subsection{Motivation and Threat Model}

\textbf{Motivation}: To enable easy participation in home automation, TAP rules are designed in an accessible form so that anyone without programming experience can easily get started.
However, this ease of use eliminates the complex modeling of logical and physical spaces, leaving users prone to configure rules incorrectly.
As a result, vulnerabilities in rule interactions are common during rule setting and may lead to unintended safety issues.
There are two types of methods to address this issue: dynamic rule enforcement control and static rule syntax correction.
However, considering complexities like latencies and physical interactions, these methods may not always prevent or rectify vulnerabilities effectively as expected.

\begin{figure}
\begin{center}
\setlength{\abovecaptionskip}{0cm}
\setlength{\belowcaptionskip}{-0.7cm}
\begin{minipage}{1\columnwidth}
  \centering
    \subfloat[Latency]{\includegraphics[width=0.49\columnwidth]{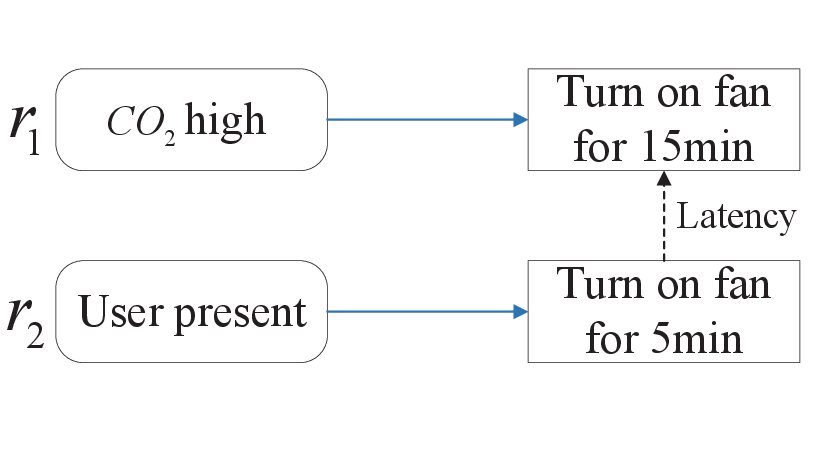}}
    \subfloat[Physical interaction]{\includegraphics[width=0.49\columnwidth]{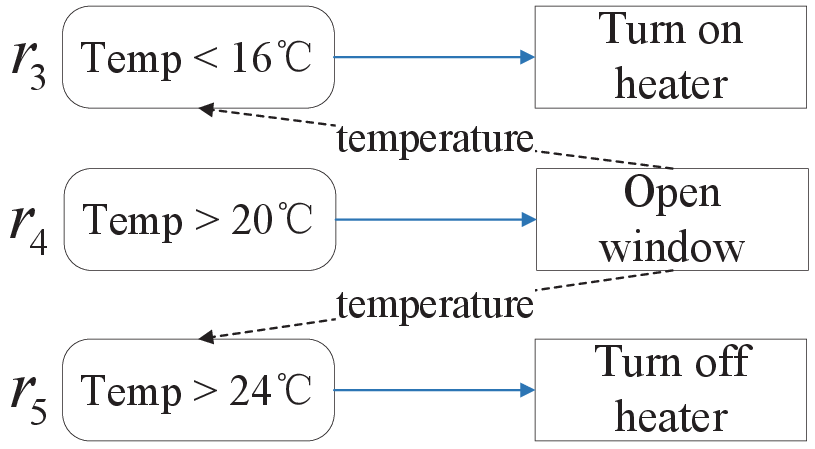}}
\end{minipage}
\vspace{1mm}
\caption{Motivation examples of rule interactions.}
\vspace{-2mm}
\label{fig:rin}
\end{center}
\end{figure}

\begin{figure*}[htbp]
\centering
\setlength{\abovecaptionskip}{0cm}
\setlength{\belowcaptionskip}{-0.4cm}
\includegraphics[width=1\textwidth]{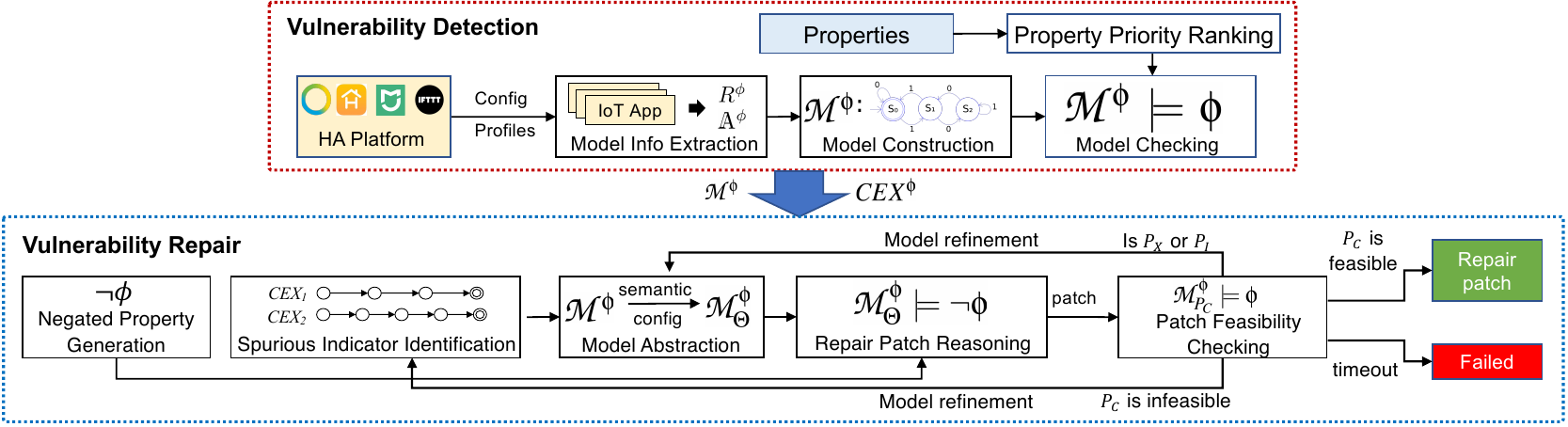}
\caption{Overview of TAPFixer.}
\label{fig:system}
\vspace{-2mm}
\end{figure*}

We show two examples as our motivation: In Fig. \ref{fig:rin}(a): a violation occurs between $r_1$ and $r_2$: if the indoor \ce{CO2} level rises up ($>$1000 ppm) within 5 minutes after the user returns home, $r_2$ will complete first and thus interrupt the ventilating fan in $r_1$, causing \ce{CO2} remaining at a high level; In Fig. \ref{fig:rin}(b): due to different indoor and outdoor temperatures, opening the window can cause the indoor temperature to drop below 16\textcelsius\ and fail to heat up to the desired 24\textcelsius. This physical interaction can cause $r_3$ to be accidentally triggered or $r_5$ to be blocked, which may lead to a fire hazard due to the heater being on the time while no one is at home. Dynamic methods tend to prevent these defects using access control (e.g., invalidating runtime effects of the fan shutdown in $r_2$ and predicting temperature changes to prevent unexpected enforcement of $r_3$ and $r_5$ in advance). They do not eliminate the root causes of defects and will face the running overhead of dynamic monitoring or prediction. In contrast, static methods fix these defects by modifying or patching rules. However, existing methods suffer from poor repair quality. For example, they can fix the violation in Fig. \ref{fig:rin}(a) by adjusting the delay or turning the fan back on after $r_1$ is interrupted, but $r_1$ is still interrupted. The major reason is that they ignore the root cause of vulnerabilities both in logical and physical spaces.

\textbf{Threat model}: We assume that rule interaction vulnerabilities are induced in three aspects:
\revise{1) \textit{misconfiguration by users}: the lack of safety practices for HA users can make it challenging to ensure global safety of installed rules, particularly when multiple family members share the HA system; 2) \textit{misleading apps}: ambiguous or incorrect app descriptions may cause users to misinterpret app functionality, leading to rule conflicts. More seriously, attackers can trick users in this way to install malicious apps they published, thereby introducing specific rules to trigger interactive threats  \cite{Chi_Zeng_Du}; 3) \textit{side-channel attacks}: an attacker can infer user and device activities by sniffing network traffic \cite{wan2022iotmosaic} and physical changes \cite{birnbach2019peeves} to exploit specific data transmission delay or physical interference to launch unsafe rule interactions. Our TAPFixer mainly focuses on the first two aspects, but also supports checking the third threat by introducing arbitrary rule delay and nondeterminacy into rule models (see $\S$\ref{subsec:modeling}).}

\vspace{-3mm}
\subsection{TAPFixer Design Intuition}
To secure TAP-based HA systems, we present TAPFixer, an automatic framework to detect and repair rule interaction vulnerabilities, as shown in Fig. \ref{fig:system}. TAPFixer takes HA apps (e.g., SmartApp, IFTTT applets) and corresponding configurations installed in the HA platform (\textit{i.e., profiles}) as input, verifies their \revise{correctness} via model checking \revise{with a set of predefined correctness properties, and generates repair patches for identified vulnerabilities. In the detection phase (see $\S$\ref{sec:verif}),}
given a correctness property $\phi$, TAPFixer extracts modeling information from profiles, formally models rule interactions ($\mathcal{M}^{\phi}$), and performs model checking ($\mathcal{M}^{\phi} \models \phi$) to identify the potential counterexample $CEX^{\phi}$ (\textit{i.e.}, a rule interaction vulnerability). $CEX^{\phi}$ is an execution path that can reach the incorrect state space \revise{against $\phi$ (\textit{i.e.}, $\neg\phi$-space shown in Fig. \ref{fig:insight}). Our detection component is designed based on many existing advanced effects \cite{chi2020cross, alhanahnah2020scalable, yu2022tapinspector, ding2021iotsafe}, but it can model and analyze TAP rules with more logical and physical features (see Table \ref{tab:detect-work} and \ref{tab:compare}) to achieve accurate vulnerability detection.}

\begin{figure}[t!]
    \centering
    \includegraphics[width=1\columnwidth]{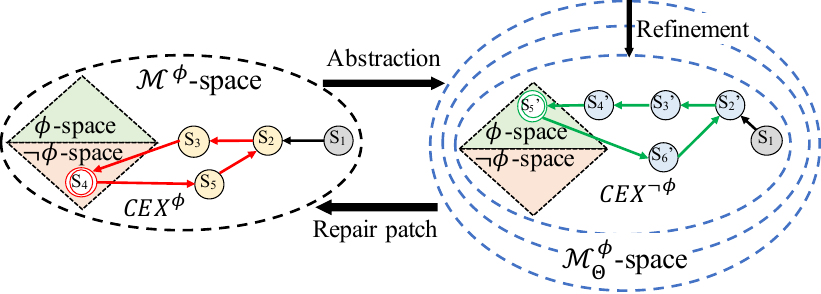}
    \caption{The core idea of our \revise{NPR} algorithm for vulnerability repair. Model abstraction via interpolation is used to involve a larger state space for patch searching, so the negated counterexample $CEX^{\neg\phi}$ can contain repair patches for the vulnerability $CEX^{\phi}$.}
    \label{fig:insight}
    \vspace{-5mm}
\end{figure}

To repair $CEX^{\phi}$, we design a novel \textit{negated-property reasoning} (NPR) algorithm in TAPFixer to generate rule patches. \revise{Its core idea (see Fig. \ref{fig:insight})} is to perform model checking with the negated property $\neg\phi$ (\textit{i.e.}, $\mathcal{M}^{\phi}\models\neg\phi$) to generate possible repair patches.
However, it is challenging since (1) $\mathcal{M}^{\phi}$ is flawed and includes much invalid information; (2) \revise{$\mathcal{M}^{\phi}$'s state space may not contain} effective repair clues.
To this end, we design an adaptive model abstraction and refinement framework (similar to the CEGAR algorithm \cite{clarke2000counterexample}) in NPR to derive the accurate patch. \revise{Its basic workflow is to abstract $\mathcal{M}^{\phi}$ to compute an over-approximated model $\mathcal{M}^{\phi}_{\Theta}$,  perform model checking with the \emph{negated property} $\neg\phi$ to generate negated counterexamples $CEX^{\neg\phi}$, refine $\mathcal{M}_{\Theta}^{\phi}$ both on rule semantics and configurations to eliminate infeasible $CEX^{\neg\phi}$, and construct patches from feasible $CEX^{\neg\phi}$ and $\mathcal{M}_{\Theta}^{\phi}$ (see $\S$\ref{sec:npr})}.

Our key insight of using model abstraction and refinement for patch generation is that abstraction can extend a larger model state space to introduce more possible repair clues and refinement guided by infeasible negated counterexamples can eliminate impossible and unrelated state transitions. An $CEX^{\neg\phi}$ identified from $\mathcal{M}^{\phi}_{\Theta}$ can enter the $\phi$-space which provides a possible path to patch the original model $\mathcal{M}^{\phi}$ to prevent it from entering into the $\neg\phi$-space. 
Taking $\mathcal{M}^{\phi}$ in Fig. \ref{fig:rin}(b) with $\phi$ (``if no one is at home, the heater should be turned off'') for example, by performing NPR with $\neg\phi$ (``if no one is at home, the heater should be turned on''), we can obtain an interpolated model $\mathcal{M}_{\Theta}^{\phi}$ and $CEX^{\neg\phi}$ which include a new execution path (``turn the heater off when the user leaves'') and a new condition constraint (``user is at home'') on $r_3$. This path and constraint in $\mathcal{M}_{\Theta}^{\phi}$ make it violate $\neg\phi$, but can construct as the repair patch of $CEX^{\phi}$.

\vspace{-4mm}
\section{Rule Modeling and Vulnerability Detection}
\label{sec:verif}

\vspace{-2mm}
\subsection{Modeling Information Extraction}
\label{subsec:extract}

To formally model TAP rules, TAPFixer first extracts device information from device platform websites (e.g., SmartThings \cite{capabilities}), transforms different implementations of automation apps into the format of TAP rules $R$, and identifies the rule set $R^{\phi}$ and attribute set $\mathbbm{A}^{\phi}$ associated with a given property $\phi$.
There are two mainstream implementation methods of the TAP paradigm: general-purpose programming languages (e.g., Groovy in SmartThings apps) and natural languages (e.g., IFTTT applets). Following existing literature \cite{chi2020cross,alhanahnah2020scalable,wang2019charting}, TAPFixer uses static program analysis and natural language processing (NLP) techniques to transform these apps into a set of TAP rules $R$. Details are described in Appendix \ref{asec:mie}.

Given a set of properties, TAPFixer filters out relative entities from the extracted rule set to compact the model through possible logical or physical interactions in a recursive way. For logical interactions, it first extracts the set of entity attributes $\mathbbm{A}^{\phi}$ presented in the property $\phi$, extracts rules $R^{\phi}$ containing attributes in $\mathbbm{A}^{\phi}$ from the rule set $R$, and appends new attributes in $R^{\phi}$ to $\mathbbm{A}^{\phi}$ accordingly. For physical interactions, \revise{we consider 9 frequently used or safety-sensitive \textit{physical channels}, mainly referring to \cite{ozmen2022discovering}}: temperature, illuminance, motion, smoke, humidity, CO, \ce{CO2}, sound, and weather status. \revise{Since TAPFixer is performed in a static manner, it cannot predict interaction behaviors in these channels dynamically. Inspired by \cite{yu2022tapinspector}, we manually conduct a qualitative analysis of device effects on physical channels in a few fixed home environments and quantify these effects in Appendix Table \ref{tab:pc} based on existing results \cite{birnbach2019peeves, yu2022tapinspector, ding2021iotsafe} and our collected sensor data. Note that the accuracy of this manual analysis may reduce as the home environment changes. Hence, to control model errors, we set these quantitative effects to a value range rather than a single value according to our measurement and also model the impact of non-deterministic interferences to a certain extent (see $\S$\ref{subsec:modeling})}.
With these quantified physical interaction effects, TAPFixer further appends both $R^{\phi}$ and $\mathbbm{A}^{\phi}$ in the same way.
This process is executed recursively until $R$ is empty or no elements can be appended to $\mathbbm{A}^{\phi}$.

\vspace{-2mm}

\subsection{Rule Interaction Modeling}
\label{subsec:modeling}

Utilizing extracted $R^{\phi}$ and $\mathbbm{A}^{\phi}$, TAPFixer models rule interactions as a finite automaton (FA). To enhance the accuracy and integrity of modeling, we introduce practical dynamic features to better align with the real world as follows:

\textbf{Latency Modeling}: we involve latency as a key element of rule syntax since a rule may be not executed immediately due to different delays \cite{chi2022delay}: (1) the delay $l_1$ defined in rules for specifying the execution time (e.g., turn lights on for 10min) or postponing action execution (e.g., the \emph{runIn} function in SmartThings). The key to modeling $l_1$ is to determine when the delay is completed. TAPFixer both considers explicit and implicit $l_1$: the former is modeled as a timer variable configured with a specified timeout value (e.g., for ``turning on the fan for 5min'', the timer has a timeout value of 300); the latter has no explicit timer, but a Boolean variable $wait\_trigger$ to denote if the targeted status is present after an extended action starts; (2) the delay $l_2$ on a \emph{tardy attribute} change to a certain value, e.g., time spent on the temperature to 20\textcelsius. It is the key factor affecting physical interactions. We discuss it in the next aspect; (3) the platform delay for data updating from devices to the platform \cite{mi2017empirical}. It is caused by updating intervals of sensors, within which a physical entity attribute $\mathbbm{A}_{phy}$ may be inconsistent with its mirror logical variable of $\mathbbm{A}_{log}$ in the platform. We observe that the updating interval varies from a few seconds (e.g., present sensors) to a dozen minutes (e.g., temperature sensors). TAPFixer sets updating interval variables for $l_3$ of different sensors to model hysteretic updates.

Given the above latency, we further formalize $r$ as follows:
\vspace{-2mm}
\begin{equation}
r_{immd} := (t, c) \stackrel{l_1}{\longmapsto} a_{immd},
\end{equation}
\begin{equation}
r_{ext} := (t, c) \stackrel{l_1}{\longmapsto} a_{ext} \stackrel{l_2}{\longmapsto} a_\Box,
\end{equation}
\begin{equation}
\mathbbm{A}_{phy} \stackrel{l_3}{\longmapsto} \mathbbm{A}_{log},
\end{equation}

\noindent where $r_{immd}$ is an immediate rule, $r_{ext}$ is an extended rule, $a_{immd}$ is an immediate action, $a_{ext}$ is an extended action, $a_\Box$ is the completion of $a_{ext}$.

\textbf{Physical Interaction Modeling}: Device actions can interact with environmental attributes via \textit{physical channels}. This physical interaction extends the way rules interact. We model physical channels based on the physical effects of immediate and tardy attributes and the three following realistic dynamic physical features (listed in Appendix Table \ref{tab:pc}):

(1) \textit{Implicit physical effect}: in addition to explicit effects of device actions, there are also implicit physical effects. For instance, opening a window can explicitly accelerate air circulation and implicitly affect the indoor temperature. We map device capabilities to physical channels they implicitly affect for rule modeling;

(2) \textit{Joint physical effect}:
devices in the same room or sharing the same physical channel may cause joint physical effects at proper distances \cite{ding2021iotsafe}. In this case, the actual effect of joint operations $\land_{i=1}^n a^{(i)}$ is the combination of individual effects. Hence, for each physical channel, we filter out devices affecting it, enumerate their co-execution scenarios, and model the sum of attribute value changes as the joint physical effect.

(3) \textit{Nondeterminacy}: there are non-deterministic physical characteristics that may vary from the case, e.g., the fire intensity, and differences in device execution due to the battery drain. Owing to this dynamic variability, TAPFixer encodes such factors as a range of randomly selected values rather than specific values, e.g., the timing threshold for evaluating the water valve to eliminate fires is set in a random way.

Finally, we formulate the \textit{physical interaction} as follows:
\vspace{-2mm}
\begin{equation}
\land_{i=1}^n a^{(i)} \stackrel{l_2}{\hookrightarrow} \mathbbm{A},
\end{equation}
where $\hookrightarrow$ denotes a physical interaction, $\land$ denotes a joint physical relation, and $l_2$ is the physical effect delay. If $\mathbbm{A}$ is an immediate channel attribute, $l_2$ is equal to 0; otherwise, we set $l_2$ according to Appendix Table \ref{tab:pc} for tardy attributes.

\textbf{Model Construction}: Given $R^{\phi}$, $\mathbbm{A}^{\phi}$, and these above models, TAPFixer translates $\mathbbm{A}^{\phi}$ as automaton variables and constructs the FA $\mathcal{M}^{\phi} := \{S, I, \Sigma \}$, where $S$ is a finite set of all states, $I\subseteq S$ is the initial state set, and $\Sigma\subseteq S \times S$ is a set of state-to-state transitions. A state $s_i\in S$ represents as a set of values for all automaton variables. Taking $r_2$ in Fig. \ref{fig:rin}(a) for example, $S$ is permutations between \{user.present, user.not\_present\} and \{fan.on, fan.off\}, and $I$ is a state \revise{arbitrarily} selected from it.

TAPFixer formulates rule executions, environmental changes, and physical interactions as $\Sigma$.
For a rule $r_k:=\langle t_k,c_k\rangle\mapsto a_k$ and the environmental attribute set $\mathbbm{A}$, a transition function $\varphi(s_{i},s_{j})\in \Sigma$ is defined as follows:
\vspace{-2mm}
\begin{equation}
\label{eq6}
\varphi(s_{i},s_{j})=\left\{
\begin{aligned}
s_{i}\times (s_{j} \hookleftarrow a_k(s_i)) & , &(t_k(s_i) \wedge c_k(s_i)) \vee \mathbbm{A}(s_i) \\
s_{i}\times s_{i} & , &otherwise.
\end{aligned}
\right.
\end{equation}
where $s_{i}$ and $s_{j}$ are the current and next automaton state, respectively, $\hookleftarrow$ means a dependency relationship caused by both logical and physical interactions, $t_k(s_i)$, $c_k(s_i)$, $a_k(s_i)$, and $\mathbbm{A}(s_i)$ are the predicate for $t_k$, $c_k$, $a_k$, and $\mathbbm{A}$ in the state $s_{i}$, respectively. $\mathbbm{A}(s_i)$ is the prerequisite that represents these environmental attributes that can naturally change (e.g., rainy conditions). The state $s_i$ will transfer to $s_j$ by action executions if both $t_k(s_i)$ and $c_k(s_i)$ are satisfied or $\mathbbm{A}(s_i)$ is true; otherwise, it remains unchanged. For instance, $\Sigma$ is the transition between the state \{user.not\_present, fan.off\} and \{user.present, fan.on\} that corresponds to $t_2$, $c_2$, $a_2$ in $r_2$.

The rule model may involve a large number of variables, some of which have large ranges (e.g., indoor illuminance can be 0 to 2000 lux), resulting in a large state space for model verification. \revise{We note that given a set of rules, an attribute’s values are activated on a limited range. Values outside this range will not affect rules' activities. Hence, to avoid state explosion, we collect all configurations defined in rules and globally compress the range of all model variables to an optimized smaller range. For example, the concerned value set of temperature in Fig. \ref{fig:rin}(b) is $\{16,20,24\}$, so we can have a smaller range $[15, 25]$ rather than the full value range of temperature. Besides, TAPFixer now has not considered floating-point numbers and it performs rounding operations on floating-point values before value compression. This is because we note that HA users are insensitive to them and generally customize rules with low-precision values, and most home sensors typically only provide coarse-grained measurement values.}

\vspace{-3mm}
\subsection{Property Specification and Prioritization}
\label{subsec:sppr}

\revise{TAPFixer uses correctness properties for vulnerability detection. A property is a criterion to describe what automation behavior is safe or not.
Generally, it can be expressed in linear temporal logic (LTL) \cite{Pnueli_1977} which describes the relative or absolute order of behaviors in the system (e.g., the next state denoted by \textbf{X}, the subsequent path denoted by \textbf{F}, and the entire path denoted by \textbf{G}). Based on properties previously defined in \cite{alhanahnah2020scalable, Celik_McDaniel_Tan_2018, chi2020cross, hsu2019safechain,wang2019charting,yu2022tapinspector,zhang2019autotap}, we select and refine them according to safety-sensitive and commonly used devices and supply more properties for different scenarios (e.g., properties with permitted latencies). We finally conduct 53 properties for vulnerability detection as shown in Appendix Table \ref{tab:sp}.}

\revise{Additionally, TAPFixer also allows users to specify their desired correctness properties.}
To facilitate users' expression of properties, there are various natural language templates \cite{zhang2019autotap}, e.g., $[state]$ should $[always]$ be active, $[event]$ should $[never]$ happen if $[state_1, ..., state_n]$. However, from the perspective of LTL formula, different templates may be equivalent in logic. Hence, we conduct four logical equivalence relations to group 9 types of language templates into 2 types of logical templates for property specification: Event-based and State-based properties, as shown in Table \ref{tab:tspnp}. While the former focuses on identifying and handling exceptions timely, the latter focuses on continuously preventing exceptions from occurring, which are often combined to ensure safety.
Based on generalized templates, we divide properties into \emph{pre-proposition} and \emph{post-proposition} using the implication symbol $\Rightarrow$ to specify event- and state-based negated properties.
Besides, we introduce the variable recording the previous state based on the jump features of events to distinguish between event and state in the model.
We discuss these relations and language templates in Appendix \ref{asec:tesp} and Table \ref{tab:lqs} due to space limitation.

\renewcommand{\arraystretch}{1.3}
\begin{table*}[htb]
\setlength{\abovecaptionskip}{0cm}
\setlength{\belowcaptionskip}{-0.2cm}
\caption{Templates of correctness property and negated property.}
\label{tab:tspnp}
\vspace{-3mm}
\begin{center}
\resizebox{\textwidth}{!}
{
\begin{tabular}{|c|c|c|c|}
\hline
Property Type &Natural Language Template &Property LTL Template & Negated Property LTL Template\\ \hline
Event-based       & IF $[state_1, ..., state_n]$, $[event]$ should $[always]$ happen.        &   G($\land_{i=1}^n state_{i}\Rightarrow X(event)$)   &   G($\land_{i=1}^n state_{i}\Rightarrow \neg X(event)$)           \\ \hline
State-based    & WHEN $[state_1, ..., state_n]$, $[state]$ should $[always]$ be active.           &  G($\land_{i=1}^n state_{i}\Rightarrow state$) &   G($\land_{i=1}^n state_{i}\Rightarrow \neg state$) \\ \hline
\end{tabular}
}
\end{center}
\vspace{-8mm}
\end{table*}

\revise{Users can both pick up concerned properties in Table \ref{tab:sp} and also specify their properties using the above 2 templates.}
Given a set of correctness properties, they may conflict with each other if they share the same device capabilities.
So, it may be impossible to secure all properties for the same rule set. An example of conflicting properties is ``close windows when it rains'' and ``open windows when \ce{CO} is detected''. TAPFixer resolves property conflicts by priority ranking, which ensures that the property with a higher priority can be fixed first, even if a conflict exists.
We formulate two types of priorities (\emph{pre-proposition} and \emph{post-proposition}) based on the composition of properties: the former mainly describes the priority of environment entities, e.g., \ce{CO} has a higher priority than rainy weather; the latter mainly describes the priority of device capabilities, e.g., locking the door has a higher priority than opening it. For properties that share the same device capability, TAPFixer prioritizes them based on pre-proposition priority and further sorts them based on post-proposition priority if they have the same priority. Properties with the same final priority are assigned in a random order. We list all correctness properties we used in Appendix Table \ref{tab:sp}.

\vspace{-2mm}
\subsection{Vulnerability Detection}

Given an LTL property $\phi$ and a rule interaction model $\mathcal{M}^{\phi}$, TAPFixer detects rule interaction vulnerabilities via model checking $\mathcal{M}^{\phi}\models \phi$. If there is a vulnerability, the model checker will output a counterexample $CEX^{\phi}$, which is a sequence of automaton states in the FA as follows:
\vspace{-2mm}
\begin{equation}\label{equ:cex}
CEX^{\phi} := <(s_{1}, \varphi_{1}), ..., (s_{k}, \varphi_{k})>.
\end{equation}
\noindent A $CEX^{\phi}$ contains a state $s^{\nvDash}$ violating against $\phi$, \textit{i.e.}, $s^{\nvDash}$ reaches the $\neg\phi$-space (see Fig. \ref{fig:insight}, e.g., the window opening when it rains).
\revise{For example, Fig. \ref{fig:exa} illustrates $CEX^{\phi}$ of Fig. \ref{fig:rin}(a) and (b) against the property P.34 with the permitted time be 10 min and P.22, respectively: in Fig. \ref{fig:exa}(a), the proposition corresponding to the violating state $Fan=on$ against P.34 is false since the fan turns off without running for at least 10 min; in Fig. \ref{fig:exa}(b), $Presence=not\_present$ is satisfied but $Heater=off$ in the next state is not for P.22}.

\begin{figure}
\begin{center}
\setlength{\abovecaptionskip}{0.1cm}
\setlength{\belowcaptionskip}{0cm}
\begin{minipage}{1\columnwidth}
  \centering
    \subfloat[P.34 $:=$ G($\ce{CO2}.high \land r_1.timer \geq 300\Rightarrow Fan.on$)]{\includegraphics[width=1.0\columnwidth]{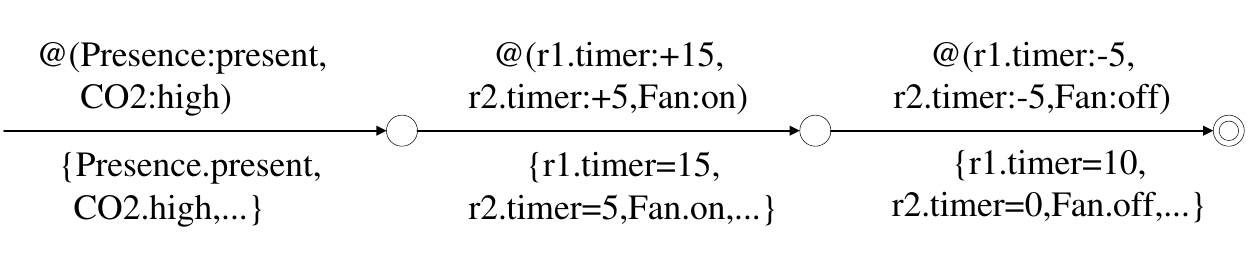}}
    \vspace{2mm}
\end{minipage}
\begin{minipage}{1\columnwidth}
  \centering
    \subfloat[P.22$:=$ G($Presence.not\_present\Rightarrow X(Heater.off)$))]{\includegraphics[width=1.0\columnwidth]{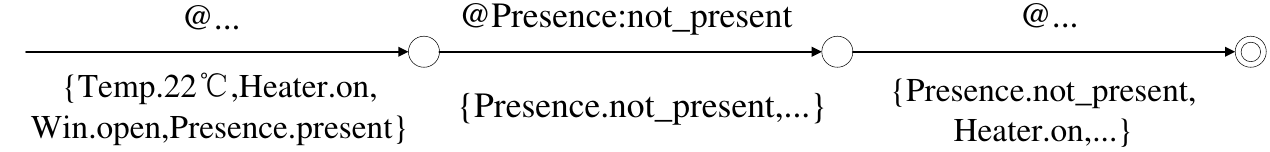}}
\end{minipage}

\caption{Example of counterexamples and their violating state. We use @ to denote the event and omit entity attributes that do not change or are irrelevant to the violation. The state with double circles denotes the violating state.}
\label{fig:exa}
\vspace{-9mm}
\end{center}
\end{figure}

\vspace{-2mm}
\section{Rule Interaction Vulnerability Repair}
\label{sec:npr}

To fix rule interaction vulnerabilities, we need a patch that can both fix logical errors and change violating physical effects. In this section, we design a negated-property reasoning (NPR) algorithm in TAPFixer for vulnerability repair. Once TAPFixer finds a counterexample $CEX^{\phi}$ for the property $\phi$ in the model $\mathcal{M}^{\phi}$, it generates the negated property $\neg\phi$ (see $\S$ \ref{subsec:ssr}) and performs NPR to generate a repair patch shown in Alg. \ref{alg:npr}. Given $CEX^{\phi}$, $\phi$, $\neg\phi$, $\mathbbm{A}^{\phi}$, $rnd=0$, and $\mathcal{M}^{\phi}$ as input, NPR performs a recursive process similar to counterexample-guided abstraction and refinement (CEGAR) to find a globally feasible repair patch $P$: it first analyses $CEX^{\phi}$ to obtain the spurious indicator $I_s$ (Line \ref{algl:ssi}, see $\S$\ref{subsec:ssr}), abstracts $\mathcal{M}^{\phi}$ as an over-approximated model  $\mathcal{M}_{\Theta}^{\phi}$ and reasons $P$ using $\mathcal{M}_{\Theta}^{\phi}\models\neg\phi$ (Line \ref{algl:abm}-\ref{algl:rps}, see $\S$\ref{subsec:regar}). Then, NPR performs local and global feasibility checking to validate the effectiveness of $P$ for fixing $CEX^{\phi}$ (Line \ref{algl:as}-\ref{algl:pnpr}, see $\S$\ref{subsec:pco}). If feasible, TAPFixer maps predicates in $P$ into corresponding rules in $R^{\phi}$ to fix $CEX^{\phi}$ (see $\S$\ref{subsec:rr}). \revise{Take Fig. \ref{fig:rin}(b) for example, NPR first extracts $I_s$ from $CEX^{\phi}$ in Fig. \ref{fig:exa}(b) (\textit{i.e.}, the violating state ``the heater is on when the user leaves'') as the indicator. Then, it collects entities in $\mathcal{M}^{\phi}$ and P.22, expands the state space of $\mathcal{M}^{\phi}$ by enlarging its predicate space with non-appearing but possibly-relative entities (e.g., the user's present status) to construct $\mathcal{M}_{\Theta}^{\phi}$. With $\mathcal{M}_{\Theta}^{\phi}$ and the negated P.22 \revise{(if the user is not at home, the heater should be turned on)}, NPR reasons a patch solution $P$ to $CEX^{\phi}$ which includes creating a new rule (``if the user is not present then turn off the heater'') and a condition (``the user is at home'') set into $r_3$.}

\setlength{\textfloatsep}{0pt}
\begin{algorithm}[!t]
\caption{Negated-property Reasoning (NPR)}
\linespread{0.9}\selectfont
\small
\label{alg:npr}
\KwIn{Counterexample $CEX^{\phi}$; Property $\phi$ and negated property $\neg\phi$; List of related entity attributes $\mathbbm{A}^{\phi}$; Number of recursive executions $rnd$; Model $\mathcal{M}^{\phi}$;}
\SetKwInOut{Output}{Output}
An empty patch $P$;\\

\If{$CEX^{\phi}$}{
    $I_s\leftarrow$ searchSpuriousIndicator($CEX^{\phi}$, $\phi$);\\\label{algl:ssi}
     $P_{I}, P_{X},iter \leftarrow CEX^{\phi}$, $CEX^{\phi},1$;\\
    \While{$iter$\,\textless\,$ITER\_LIMIT$}{
        $\mathcal{M}_{\Theta}^{\phi}, \leftarrow$\,abstractModel($P_I$,$P_X$,\:$\phi$,\:$\mathbbm{A}^{\phi}$,\:$\mathcal{M}^{\phi}$);\\\label{algl:abm}
        $P$ $\leftarrow$\,reasonPatch($\mathcal{M}_{\Theta}^{\phi}$, $\neg\phi$);\\\label{algl:rps}
        $P_C, P_I, P_X \leftarrow$checkLocalFeasibility($P$,\:$\phi$,\:$I_s$,\:$CEX^{\phi}$);\\\label{algl:as}
        \tcp{classify $P$ as $P_C$, $P_I$ or $P_X$}
        \If{$P_C$}{$CEX_{P_C}^{\phi},\!\mathcal{M}_{P_C}^{\phi}\!\leftarrow$\!verifyGlobalFeasibility(\!$\phi$,\,$\mathcal{M}^{\phi}$,\,$P_C$);\label{algl:vg}\\
            \If{$CEX_{P_C}^{\phi}$}{
                $rnd \leftarrow rnd+1$;\\\label{algl:rnd}
                \If{$rnd <\,$ROUND\_LIMIT}{
                $P\leftarrow$ NPR($CEX_{P_C}^{\phi}$, $\phi$, $\neg\phi$,$\mathbbm{A}^{\phi}$, $rnd$, $\mathcal{M}_{P_C}^{\phi}$);\label{algl:pnpr}
                }
            }\lElse{
            \textbf{Break}
             }
        }\lElse{$iter \leftarrow iter+1$}
    }
}

\Output{$P$}
\end{algorithm}
\setlength{\textfloatsep}{0pt}

\vspace{-3mm}
\subsection{Negated Property Generation and Spurious Indicator Identification}
\label{subsec:ssr}

The core method of NPR is to use the property $\neg\phi$ negated to the property $\phi$ to reason about the repair patch for the identified $CEX^{\phi}$. With our two merged LTL property templates, we give their negated templates in Table \ref{tab:tspnp} and generate $\neg\phi$. To prove the soundness of property negation, we manually prove these negations shown in Appendix \ref{asec:tesp}.

Besides, the key indicator of whether a patch can fix the vulnerability $CEX^{\phi}$ is the ability to eliminate the violating state $s^{\nvDash}$ in $CEX^{\phi}$ \revise{(e.g., ``the heater is on when the user leaves'' in Fig. \ref{fig:exa}(b))}. Hence, we define $s^{\nvDash}$ as the \textit{spurious indicator} $I_s=s^{\nvDash}$ to assess the ability of repair patches
In many cases of $CEX^{\phi}$, $s^{\nvDash}=\varphi_{k}(s_k)$. But, there are still some cases: $CEX^{\phi}$ is a lasso-shaped path from $s_1$ which has a cycle (e.g., Fig. \ref{fig:insight}), so $\varphi_{k}(s_k)\neq s^{\nvDash}$. Hence, NPR further analyses $CEX^{\phi}$ to find  $I_s$. Based on the given property $\phi$ and $CEX^{\phi}$, NPR converts the problem of identifying $I_s$ into the propositional logic satisfiability problem.
Specifically, NPR evaluates each automaton state as a judging unit and analyzes each of its entity values based on $\phi$'s LTL encoding \cite{Biere_Cimatti_Clarke_Zhu_1999}. If an entity in a state is unsatisfactorily encoded, this state is considered as $I_s$.

\vspace{-2mm}
\subsection{Patch Reasoning}
\label{subsec:regar}

The core problem of NPR is to generate an abstract model $\mathcal{M}_{\Theta}^{\phi}$ via model abstraction and refinement, with which NPR can reason a global feasible patch via $\mathcal{M}_{\Theta}^{\phi}\models\neg\phi$. Both constraints and assignments in $\mathcal{M}^{\phi}$ can be expressed by predicates whose entity attributes and values represent model semantics and configurations, respectively.
We define two types of predicates: \textbf{\textit{status predicate}} (denoted by $takeValue(x)$) depicts ``an entity attribute $x$ takes the value of $Value$'', where $x$ is the semantic and $Value$ is the configuration; \textbf{\textit{trigger predicate}} (denoted by $isTrigger(y)$) depicts ``the constraint related to an entity attribute $y$ is the state trigger''.
Note that the trigger predicate is necessary for a rule to represent its existence in the model.
A TAP rule can be expressed as a compound proposition in the form of predicates in the model.
For example, $r_3$ in Fig. \ref{fig:rin}(b) can be expressed as $isTrigger(Temp) \land takeBelow\,16\textcelsius(Temp) \mapsto takeOn(Heater.switch)$.
Based on predicates, we now discuss our design of model abstraction via interpolation in terms of semantics and configuration, repair patch reasoning, and reasoning-guided abstraction refinement in NPR as follows:

\textbf{Model Semantic Abstraction}.
\label{rule semantics abstract}
Abstracting model semantics via predicate interpolation can introduce new behaviors to the current state space for patch construction.
Given the model $\mathcal{M}^{\phi}$ related to a property $\phi$, NPR extracts predicates from model transitions to merge as a universal predicate set $U$ and defines a predicate set $\Theta$ for semantic interpolation, whose initial items are predicates defined in $\phi$. To perform semantic interpolation, NPR chooses these transitions in $\mathcal{M}^{\phi}$ that contain these semantics used in $\Theta$ as interpolation targets. We discuss semantic interpolation in two defective scenarios.

For the defective scenario that requires perfecting existing rules, repairing is the problem of perfecting status predicates in all model states. NPR interpolates each state $s_i$ in $\mathcal{M}^{\phi}$ with the difference predicate set $U-s_i$, so as to provide a full model space for predicate modification or generation based on new semantics. Take $r_3$ in Fig. \ref{fig:rin}(b) for example, its $U$ and $S$ consists of predicates on attributes $\{Temp, Presence, Heater.switch\}$ and $\{Temp, Heater.switch\}$, respectively. Hence, its $\Theta$ consists of predicates on $\{Presence\}$, and these predicates are connected by disjunction $\lor$.
TAPFixer interpolates $r_3$ as follows:

\vspace{-4mm}

\begin{small}
\begin{align*}
&isTrigger(Temp)\land takeBelow\,16\textcelsius (Temp)\land (takeOn(Presence)\\
&\lor takeOff(Presence))\mapsto takeOn(Heater.switch)
\end{align*}
\end{small}
\vspace{-4mm}

For the defective scenario that requires new rules, repairing is the problem of generating new transitions that have a new trigger predicate. To generate an interpolated transition, NPR includes $U$ into $\Theta$ to generate its state with a trigger and state predicate and uses the action of generating the device state desired in $\phi$ as its action. Take Fig. \ref{fig:rin}(a) for example, NPR joins predicates of the same type by disjunction $\lor$ and predicates of different types by conjunction $\land$ to construct an interpolated transition:

\vspace{-4mm}

\begin{small}
\begin{align*}
&\left(isTrigger(\ce{CO2})\lor\!isTrigger(Fan)\lor\!isTrigger(Presence)\right)\\
&\land (takeHigh(\ce{CO2})\!\lor\!takeModerate(\ce{CO2})\!\lor\!takeLow(\ce{CO2}))\\
&\land (takeOn(Fan)\!\lor\!takeOff(Fan))\land\!(takePresent(Presence)\\
&\lor takeNotPresent(Presence)) \mapsto takeOn(Fan)
\end{align*}
\end{small}

\vspace{-4mm}

Predicate interpolation can induce a huge state space and may cause a state explosion for satisfiability checking.
We observe that in fact, the satisfiability of many predicates can be denoted by one predicate due to physical constraints among them.
Hence, NPR introduces a Boolean variable to represent one or more joint predicates.
For instance, the satisfiability of $takeOn(Fan)$ is denoted by $flag\!\to\!takeOn(Fan)$ and its truth value is equivalent to $flag$.
For status predicates of the same entity attribute, NPR constrains the sum of their flags to be 1 or 0, since an entity either has only one state at a time or does not exist.
Similarly, for the trigger predicate of a TAP rule, NPR constrains its sum to be 1 or 0 since the rule has only one trigger at a time or does not exist.
Simplifying the satisfiability of predicates with a single Boolean variable can avoid state explosion and reduce model checking overhead.

\textbf{Model Configuration Abstraction}.
There are many enumerated (e.g., switch) and numerical configurations (e.g., illuminance) in the HA system. While the former has a small enumeration scope, the latter may have a large range resulting in an excessively large number of predicates. Hence, for abstracting a numerical configuration, NPR only takes a few practical values as its abstraction set for predicate generation, rather than the full value scope \revise{to reduce the state space of abstract models}: NPR first finds its variable values in rules that are logically and physically related to, as its initial abstraction set; afterward, according to the physical effect received by it in one unit of time, NPR appends the change value of each element in its abstraction set in three units of time to the set and removes duplicated values.
For instance, the initial abstraction set of temperature in Fig. \ref{fig:rin}(b) is $\{16, 20, 24\}$, and its complementary values range from 13 to 27 caused by physical effects, which is much smaller than the original scope of temperature.

Latency is a special model configuration since it contains both numerical and enumerated configurations. We observe that fixing the vulnerability by adjusting or even removing the numerical delay defined in the rules would violate the user's intent.
To be user-centric, NPR selects enumerated latency configurations (used to postpone action execution such as $runIn$ in SmartThings and $wait\_trigger$ in HomeAssistant) as the abstraction scope and follows the above abstraction method to construct their abstraction set, which preserves the desired delay. Take $r_1$ in Fig. \ref{fig:rin}(a) for example, the configuration of its action can be abstracted in the predicate form of the original delay type $take15minOn$ and other delay type $takeDelayOn$ as follows:

\vspace{-4mm}

\begin{small}
\begin{align*}
&isTrigger(\ce{CO2}) \land takeHigh(\ce{CO2})\\
&\mapsto take15minOn(Fan) \lor takeDelayOn(Fan)
\end{align*}
\end{small}
\vspace{-4mm}

\noindent For $takeDelayOn$ (\textit{i.e.}, waiting for \ce{CO2} to be no longer high, e.g., 1000 ppm), NPR sets a Boolean variable \emph{wait\_trigger} as its abstraction set. NPR performs model configuration abstraction after model semantic abstraction is finished.

\textbf{Repair Patch Reasoning}.
Through abstraction, the abstract model $\mathcal{M} _{\Theta}^{\phi}$ involves a larger state space which may encompass relevant repair patches. NPR performs the model checking ($\mathcal{M} _{\Theta}^{\phi}\models \neg \phi$) to identify a negated counterexample $CEX_{\Theta}^{\neg\phi}$ which is an execution path that can reach the secure $\phi$-space.
That means $\mathcal{M} _{\Theta}^{\phi}\models \neg \phi$ and $CEX_{\Theta}^{\neg\phi}$ may provide a potential solution for fixing $CEX^\phi$. Hence, we call $CEX_{\Theta}^{\neg\phi}$ a possible repair \textit{patch} (\textit{i.e.}, $P=CEX_{\Theta}^{\neg\phi}$). However, due to impossible paths or states introduced by abstraction, the patch may be infeasible or be able to introduce other new vulnerabilities in the original rule model. Hence, NPR validates its feasibility (see $\S$\ref{subsec:pco}), and for infeasible patches, NPR refines $\mathcal{M}_{\Theta}^{\phi}$ as follows. In addition, it is possible that there is no available patch in the current abstract state space due to inaccurate abstraction or impossible state for $\phi$. NPR will further refine the model before timing out.

\textbf{Reasoning-guided Abstraction Refinement}. The fundamental reason for infeasible patches is coarse-grained abstraction. To eliminate invalid repair patches, NPR first classifies a repair patch $P$ as $P_C$ (correct), $P_I$ (implausible), and $P_X$ (plausible but incorrect) according to feasibility checking in $\S$\ref{subsec:pco}, and refines the granularity of abstraction to exclude infeasible transitions using the invalid patch ($P_I$ or $P_X$) as follows.

To eliminate interpolations leading to $P_I$, NPR extracts the automaton subsequence before the spurious indicator and keeps certain critical automaton variables unchanged within the abstraction. Specifically, for tardy attributes and immediate attributes that are not affected by the device (e.g., rainy weather and the activation of the motion sensor), NPR collects their values in each state of the subsequence and keeps them unchanged during model refinement.
To eliminate interpolations leading to $P_X$, NPR extracts these incorrect combinations of predicates in the previous abstraction and removes them by appending an invariant to the abstract model during refinement. For instance, changing the weather to fix the vulnerability in Fig. \ref{fig:exa}(b) is plausible but incorrect. To eliminate this $P_X$, NPR adds a negation invariant $\neg \,takeNotRainy(Weather)$ in the next abstraction, which avoids the weather predicate incorrectly taking the value of not rainy.

Through iterations of model abstraction and refinement, NPR iteratively refines the abstraction granularity and eliminates invalid repair patches \revise{until a global feasible patch is found.} The procedure is repeated $ITER\_LIMIT$ times.

\subsection{Patch Feasibility Checking}
\label{subsec:pco}

Given a patch $P$, NPR needs to validate its feasibility for vulnerability repair. A straightforward method is to update the model using $P$ and verify with $\phi$. However, such a method may have an expensive checking overhead.
We observe that some patches contain impossible changes to environment constraints (e.g., weather), which can be identified locally.
Hence, we design the method of patch feasibility checking for NPR including local and global feasibility checking.

\textbf{Local Feasibility Checking}. NPR aims to perform local feasibility checking to distinguish the patch $P$ as $P_C$, $P_I$, or $P_X$. We summarize the following two interaction scenarios containing spurious constraint changes:

(1) Altering \emph{physical environment attributes unaffected by device executions} to eliminate the violating state $s^{\nvDash}$ is invalid. According to the position where environment attributes are changed in $CEX^{\phi}$, $P$ can be categorized as $P_I$ or $P_X$.
For instance, the vulnerability $CEX^{\phi}$ of windows opening during rainfall can be repaired by changing the weather to no rain in the preceding state of $s^{\nvDash}$, but this is obviously not possible. The above example belongs to $P_X$ since changing weather can eliminate $s^{\nvDash}$, but is infeasible.
Regarding $P_I$, taking Fig. \ref{fig:exa}(b) for example,
the implicit effect that opening windows will decrease the indoor temperature only occurs when the outdoor temperature is lower than the indoor temperature; otherwise, the defect rule $r_5$ will never be blocked (e.g., the outdoor temperature is 26\textcelsius). This scenario belongs to $P_I$ since it changes the fixed outdoor temperature from 12\textcelsius\, to 21\textcelsius, which is unrealistic and cannot eliminate $s^{\nvDash}$ in $CEX^{\phi}$.

(2) \textit{The non-elimination of the violating state} in logical environment $CEX^{\phi}$ is another unsatisfied case.
Vulnerabilities may not disrupt the entire rule execution and safe executions can still exist in the rest of rule interactions, even if the current vulnerability is not eliminated. Take the scenario in Fig. \ref{fig:exa}(b) for example, there may be a safe situation where rules are not intersected when the indoor temperature is 18\textcelsius\,\, and the heater is off while the user is not at home before the violation occurs, but the violating state still occurs in the same place. Such scenarios also correspond to $P_X$ since they are relevant to $\phi$ but cannot eliminate the violating state.

To identify spurious patches, NPR extracts environment constraints from the patch $P$ and compares them with those in $CEX^{\phi}$ to classify $P$ as $P_I$ or $P_X$ as follows: (1) to identify $P_I$, NPR adopts unaffected entities to determine the constraint satisfiability after abstraction since their execution paths should remain unchanged during model abstraction. Hence, NPR compares states in $P$ with these in $CEX^{\phi}$. If there are different values of any unaffected attributes, the constraint is unsatisfied and $P$ is set as $P_I$;
(2) to identify uncontrollable $P_X$, NPR figures out the transition in $P$ leading to its $s^{\neg\nvDash}$ and checks if it can be manipulated by the HA system. If not controllable, NPR sets $P$ as $P_X$; (3) to identify $P_X$ that fails to eliminate $s^{\nvDash}$, NPR compares $P$ with $I_{s}$: if $I_{s}$ exists in $P$. It means the violating state $s^{\nvDash}$ in $CEX^{\phi}$ can not be eliminated by $P$. Hence, $P$ is classified as $P_X$.

\textbf{Global Feasibility Checking.}
A correct repair patch needs to be able to fix the identified vulnerability $CEX^{\phi}$ without introducing new vulnerabilities in other rule interactions (\textit{i.e.}, global feasibility).
Given a locally feasible patch $P=P_C$, NPR further verifies its correctness in the original model $\mathcal{M}^{\phi}$. NPR updates $\mathcal{M}^{\phi}$ as a fixed model $\mathcal{M}^{\phi}_{P_C}$ by modifying existing transitions, adding new transitions, or removing deleted transitions according to transitions in $P_C$. Then, NPR performs model checking $\mathcal{M}^{\phi}_{P_C} \models \phi$.
If no counterexample, $P_C$ is globally feasible and NPR will terminate to output $P_C$; otherwise, NPR performs iterative model refinement to optimize the abstract model until a globally feasible patch is generated.
\vspace{-3mm}
\subsection{TAP Rule Repair}
\label{subsec:rr}
\revise{Unlike traditional software, repairing TAP rules may change their meanings and the effect the user desires. Hence, inspired by \cite{Chi_Zeng_Du,zhang2019autotap}, to ensure that the patch can effectively repair vulnerabilities and is user-centric, it must be (1) \textit{accommodating}: original behaviors that are verified to be globally safe should be preserved; (2) \textit{safety-compliant}: the patched TAP system should be globally correct; (3) \textit{valid}: patches should be deployable. Specifically, to be accommodating, TAPFixer retains rules without vulnerabilities instead of modifying them since abstraction on them does not eliminate the violating state and fails to pass the local feasibility check. For defective rules, TAPFixer retains them and repairs them by complementation (both for existing and new rules) rather than discarding them. Since TAPFixer applies a reasoning-guided strategy in model refinement to ensure that the constructed patch is globally feasible, $P$ is \textit{safety-compliant}. Besides, the repair process is user-centric since predefined correctness properties match user preferences and patch results are reported to users to confirm. To be \textit{valid}, TAPFixer ensures that $P$ follows rule syntax and physical constraints (see $\S$\ref{subsec:pco}).}

Specifically, once a global feasible patch $P$ is identified, TAPFixer uses $P$ to patch the original rule set $R^{\phi}$ to repair the violation $CEX^{\phi}$. It first identifies these transitions $\{\varphi_j\}_0^K$ ($K$ is the number of new transitions) in $P$ different from $\mathcal{M}^{\phi}$ to construct as the repair policy $\widehat{\Sigma}$.
Following the path of $P$, TAPFixer identifies different elements of each transition $\varphi_j$ in $\widehat{\Sigma}$ to generate updating operations to $R^{\phi}$ as follows: if $\varphi_j$ has a new trigger predicate or assignment (\textit{i.e.}, a new action) than existing transitions, it translates $\varphi_j$ into a new TAP rule and sets into $R^{\phi}$; if $\varphi_j$ has a partial different status predicate, it updates the corresponding existing rule by modifying its condition \revise{or action} according to $\varphi_j$; if $P$ skips some existing transitions, it removes corresponding TAP rules. \revise{Take Fig. \ref{fig:rin}(a) for example, the generated $\{\varphi_j\}_0^K$ contains two predicates $takeDelayOn(Fan)$ for $r_1$ and $r_2$. Hence, NPR generates a condition (``\ce{CO2} is low'') to both update the activation condition of $a_\Box$ (``turn off the fan'') in $r_1$ and $r_2$.}

\revise{Finally, these repaired rules are required to be converted into the form of corresponding HA app programs. However, there is no unified programming language used in different HA platforms, e.g., natural language in IFTTT, Groovy in SmartThings, and GUI in MI Home.
Hence, TAPFixer currently only focuses on generating repair patches in the form of rule syntax (\textit{i.e.}, rule patches) and outputting found counterexamples with corresponding properties and repair patches to users. We leave program repairs as manual work which may require users to follow the patch to reprogram HA apps using the language provided by the platform. This process is straightforward since HA apps are relatively small and simple.}

\vspace{-2mm}
\section{Implementation and Evaluation}
\subsection{Implementation}

To evaluate TAPFixer, we implement a prototype system of TAPFixer with 4216 lines of code in Python, including: (1) \emph{Modeling Information Extractor}, (2) \emph{Model Builder}, (3) \emph{Model Validator}, (4) \emph{CEX Analyzer}, and (5) \emph{Model Abstractor}.
We use a mature symbolic model checker nuXmv\cite{nuXmv} as the model-checking engine in TAPFixer so that it can perform unbounded model checking with given properties in the LTL formula to verify rule interaction vulnerabilities and bounded model checking (BMC) for repair patch reasoning to repair them. Besides, we develop 53 relevant correctness properties (shown in Appendix Table \ref{tab:sp}) based on vulnerability patterns to detect and fix the vulnerability. The source codes of our TAPFixer is available at: \url{https://github.com/q1uTr5th/TAPFixer}.

\subsection{Experiment Setup}
The TAPFixer prototype is evaluated from four aspects: \textit{case study}, \textit{vulnerability verification and repair of real-world HA apps}, \textit{user study}, and \textit{performance benchmarks}. The case study is used to validate TAPFixer's accuracy for detecting and fixing vulnerabilities in comparison to existing approaches. A set of IoT market apps is then applied to evaluate whether TAPFixer can identify and repair property violations in practice and analyze the repair success rate (RSR) achieved by TAPFixer. We next conduct a user study to investigate potential vulnerabilities introduced during rule configuration by users with varying levels of knowledge about TAP rules, so as to verify the bug-fixing capability of TAPFixer in practical IoT scenarios. Finally, we analyze the time overhead of TAPFixer in vulnerability detection and repair generation.

We evaluate TAPFixer on two most popular HA platforms SmartThings and IFTTT which have been widely studied in the existing literature. We totally obtain 1177 TAP rules from two aspects: 49 groovy apps with a total of 149 rules from a public SmartThings benchmark IoTMAL \cite{IoTBenchtestsuite} and 1028 IFTTT applets from the IFTTT dataset used in \cite{mi2017empirical}. Device capabilities provided in \cite{DeviceCapabilitiesReference} are also collected to formulate TAP rules. TAPFixer is evaluated on an HP desktop with 2.30GHz Intel Core i7-11800H and 16GB memory.

\subsection{Case Study of Accuracy}
\label{casestudy}

Our case study is conducted from two aspects: detection and repair. First, to the detection accuracy of TAPFixer, we use a detection benchmark SmartHomeBench \cite{SmartHomeBench} (extended based on IoTMAL) to compare TAPFixer with several existing methods, including SOATERIA\cite{Celik_McDaniel_Tan_2018}, SAFECHAIN\cite{hsu2019safechain}, IOTCOM\cite{alhanahnah2020scalable}, and TAPInspector\cite{yu2022tapinspector}. Note that SOATERIA and TAPInspector are closed-sourced, but we can use their report results provided in \cite{Celik_McDaniel_Tan_2018, alhanahnah2020scalable, yu2022tapinspector}. Hence, we run SAFECHAIN, IOTCOM, and TAPFixer on SmartHomeBench and show the accuracy comparison results in Table \ref{tab:accuracy_compare}. Individual apps \textit{ID-1-9} and app groups \textit{Gp-1-3} contain violations only related to rule logic. \textit{Gp-4-6} contain violations related to physical channels. \textit{N1-3} and \textit{Gp-N4-5} provides violation cases corresponding to expanded vulnerabilities. Besides, we further develop a SmartApp \textit{N-6} containing a rule ``\textbf{IF} smoke detected, \textbf{THEN} open water valve for 10min.'' with the property P.42 as a benchmark case, , which can be used to evaluate the detection capability of nondeterminacy violations.
The time threshold for evaluating the water valve to eliminate fires is determined by the fire intensity. The threshold in \textit{N-6} is preset to fixed 10 min, but it may require to 15 min to eliminate the fire due to the greater fire intensity, which causes the water valve to close after 10 min and violates P.42.

We can find that SOATERIA and SAFECHAIN fail to identify expanded vulnerabilities due to limited modeling capabilities for physical and latency factors. IOTCOM supports the analysis of physical channels, so it can detect more violations.
TAPInspector supports most cases, but not the nondeterminacy in \textit{N-6}. TAPFixer introduces more comprehensive physical and latency modeling, thus identifying all violations successfully. The fundamental reason is that existing detection methods do not support the modeling of many physical features (we discuss them in Table \ref{tab:detect-work}), leading to a low detection capability.

Existing static vulnerability repair methods focus on simple rule conflicts and lack attention to expanded complex vulnerabilities.
Hence, we do not use SmartHomeBench as the repair benchmark since it contains many simple violations. In turn, we construct a set of flawed rules containing complex vulnerabilities as the benchmark to assess differences in the repair capability between TAPFixer and existing methods.
For conciseness, we give the text-based description of these flawed rules in Table \ref{flawed}.
\emph{Group 1} contains three rules shown in Fig. \ref{fig:rin}(b) and involves a vulnerability \textbf{V4}.
In \emph{Group 2}, the delay used to turn off the AC is mistakenly configured to turn it on. It leads to the vulnerability \textbf{V5} that two rules interact in latency disorder and perform conflicting actions, which causes the AC to not turn off automatically after it is on. There is a vulnerability \textbf{V6} in \emph{Group 3} that if the user leaves less than ten minutes after returning home, the action of the first rule can override the power-off effect of the second rule.
\emph{Group 4} is the extended scenario of Fig. \ref{fig:rin}(a), where the vulnerability \textbf{V7} also occurs between the first and second rule.
\emph{Group 5} contains four rules for temperature regulation. If the temperature is below 18\textcelsius\ when the user presents, the vulnerability \textbf{V8} occurs since the third rule may be triggered before the second rule blocks its activation. We also consider physical- and latency-related initialization scenarios (\emph{N/A 1-2}) that violate the event-based property P.28 and state-based property P.21 (in Appendix Table \ref{tab:sp}), respectively.

\begin{table}[]
\caption{Accuracy comparison of the vulnerability detection. We use \truep\ ,\falsep\ , and \falsen\ to denote true positive, false positive, and false negative, respectively.}
\label{tab:accuracy_compare}
\resizebox{\linewidth}{!}
{
\begin{threeparttable}
\begin{tabular}{|c|c|c|c|c|c|}
\hline
Benchmark    & SOATERIA\tnote{*} & SAFECHAIN & IOTCOM & TAPInspector\tnote{*} & TAPFixer \\ \hline
\textit{ID-1} &    \truep\      &     \truep\      &   \truep\     &      \truep\        &     \truep\     \\ \hline
\textit{ID-2} &      \falsep\    &     \falsen\      &    \truep\    &       \truep\       &     \truep\     \\ \hline
\textit{ID-3} &     \truep\     &      \falsen\     &   \truep\     &      \truep\        &     \truep\     \\ \hline
\textit{ID-4} &   \truep\ \falsen\       &    \falsen\       &   \truep\     &        \truep\      &     \truep\     \\ \hline
\textit{ID-5} &    \falsen\      &     \truep\      &     \falsen\   &      \truep\        &    \truep\      \\ \hline
\textit{ID-6} &     \truep\     &     \falsen\      &    \truep\    &      \truep\        &    \truep\      \\ \hline
\textit{ID-7} &    \truep\     &     \truep\      &    \truep\    &       \truep\       &   \truep\       \\ \hline
\textit{ID-8} &     \truep\     &    \falsen\       &    \truep\    &       \truep\       &   \truep\       \\ \hline
\textit{ID-9} &      \falsep\    &      \truep\     &    \truep\    &      \truep\        &   \truep\       \\ \hline
\textit{N-1}  &      \falsen\    &      \falsen\     &    \falsen\    &      \truep\        &   \truep\       \\ \hline
\textit{N-2}  &     \falsen\     &     \falsen\      &      \falsep\  &     \truep\         &    \truep\      \\ \hline
\textit{N-3}  &    \falsen\      &      \falsen\     &    \falsen\    &       \truep\       &    \truep\      \\ \hline
\textit{N-6}  &      \falsen\    &    \falsen\       &   \falsen\   &        \falsen\      &   \truep\       \\ \hline
\textit{Gp-1} &     \truep\     &     \falsen\      &    \truep\    &      \truep\        &   \truep\       \\ \hline
\textit{Gp-2} &     \truep\     &     \falsen\      &    \truep\    &      \truep\        &   \truep\       \\ \hline
\textit{Gp-3} &      \truep\    &       \truep\    &    \truep\    &      \truep\        &    \truep\      \\ \hline
\textit{Gp-4} &     \falsen\     &     \falsen\      &    \truep\    &    \truep\          &    \truep\      \\ \hline
\textit{Gp-5} &     \falsen\     &      \falsen\     &    \truep\    &     \truep\         &  \truep\        \\ \hline
\textit{Gp-6} &     \falsen\     &     \falsen\      &    \truep\    &     \truep\         &    \truep\      \\ \hline
\textit{Gp-N4}  &     \falsen\     &       \falsen\    &   \falsen\     &      \truep\        &    \truep\      \\ \hline
\textit{Gp-N5}  &      \falsen\    &       \falsen\    &    \falsen\    &       \truep\       &   \truep\       \\ \hline

\end{tabular}
\begin{tablenotes}
    \item[*] results obtained from \cite{Celik_McDaniel_Tan_2018, alhanahnah2020scalable, yu2022tapinspector}
\end{tablenotes}
\end{threeparttable}
}
\end{table}

\begin{table}[t!]
\centering
\caption{Designed flawed rule groups.}\label{tab:frg}
\label{flawed}
\resizebox{\linewidth}{!}
{
\begin{tabular}{|c|c|l|}
\hline
\begin{tabular}[c]{@{}c@{}}Bench-\\mark\end{tabular} & \begin{tabular}[c]{@{}c@{}}Vulner-\\ability\end{tabular} & \multicolumn{1}{c|}{Rule Description}                                                                                                                                                 \\
\hline
\emph{Group 1}   &        \textbf{V4}                                                              &    \begin{tabular}[c]{@{}l@{}}\textbf{IF} temperature $\textless$ 16\textcelsius, \textbf{THEN} turn on heater.\\\textbf{IF} temperature $\textgreater$ 24\textcelsius, \textbf{THEN} turn off heater.\\\textbf{IF} temperature $\textgreater$ 20\textcelsius, \textbf{THEN} open window.\end{tabular}                                                                                                                                                                              \\
\hline
\emph{Group 2}   &         \textbf{V5}                                                                   &   \begin{tabular}[c]{@{}l@{}}\textbf{IF} after user present 10 min, \textbf{THEN} turn on AC.\\\textbf{IF} user is present, \textbf{THEN} turn off AC.\end{tabular}                                                                                                                                                                               \\
\hline
\emph{Group 3}   &      \textbf{V6}                                                                      &     \begin{tabular}[c]{@{}l@{}}\textbf{IF} after user present 10 min, \textbf{THEN} turn on blanket.\\\textbf{IF} user is not present, \textbf{THEN} turn off blanket.\end{tabular}                                                                                                                                                                              \\
\hline
\emph{Group 4}   &           \textbf{V7}                                                                 & \begin{tabular}[c]{@{}l@{}}\textbf{IF} \ce{CO2} \revise{> 1000 ppm}, \textbf{THEN} turn on fan for 15min.\\\textbf{IF} air \revise{humidity > 80\%}, \textbf{THEN} turn on fan for 10min.\\\textbf{IF} user is present, \textbf{THEN} turn on fan for 5min.\end{tabular}  \\

\hline
\emph{Group 5}          &     \textbf{V8}                                                                     &       \begin{tabular}[c]{@{}l@{}}\textbf{IF}  temperature $\textless$ 18\textcelsius\ \textbf{WHILE} user is present,\\ \textbf{THEN} turn on heater.\\
\textbf{IF}  temperature $\textless$ 25\textcelsius\, for 20 min \textbf{WHILE} user is \\present and window is closed, \textbf{THEN} turn off heater.\\
\textbf{IF} temperature $\textgreater$ 27\textcelsius, \textbf{THEN} open window.\\
\textbf{IF} temperature $\textless$ 16\textcelsius, \textbf{THEN} close window.\end{tabular}                                                                                          \\

\hline
\emph{N/A 1}   &           P.28 violation                                                                 & The violation of event-based P.28 with no rules.\\
\hline
\emph{N/A 2}   &           P.21 violation                                                                 & The violation of state-based  P.21 with no rules.\\
\hline
\end{tabular}
}
\end{table}

\begin{table}[t!]
\centering
\caption{Repair accuracy comparison of expanded vulnerabilities.
}
\label{otherwork1}
\footnotesize
\resizebox{\linewidth}{!}
{
\begin{threeparttable}[b]
\begin{tabular}{|c|c|c|c|c|}
\hline
Benchmark & Liang et al. \cite{liang2016systematically} & MenShen \cite{bu2018systematically} & AutoTap \cite{zhang2019autotap} & TAPFixer  \\
\hline
 \emph{Group 1}   &   \color{carnelian}\faTimesCircle\tnote{$\dagger$}    &  \falsep  & \color{carnelian}\faCircle\tnote{$\ddagger$}  &  \truep    \\
\hline
  \emph{Group 2}   &  \falsep    &  \falsep  & \falsep  &   \truep  \\
\hline
  \emph{Group 3}   &  \falsep    &  \falsep  &  \falsep  &   \truep   \\
\hline
  \emph{Group 4}   &  \falsen    &  \falsen  & \falsep  &    \truep  \\
\hline
  \emph{Group 5}   &  \falsen    &  \falsen  &  \truep   &   \truep  \\
\hline
 \emph{N/A 1}      &  \falsen    &  \falsen  &  \truep  &   \truep \\
\hline
 \emph{N/A 2}      &  \falsen    &  \falsen  & \truep  &  \truep \\
\hline
\end{tabular}

 \begin{tablenotes}
        \footnotesize
        \item[$\dagger$]  Correctly fixed partial vulnerable rule interactions, but did not fix the rest.
        \item[$\ddagger$]  Correctly fixed partial vulnerable rule interactions, but incorrectly fixed the rest.
      \end{tablenotes}
\end{threeparttable}
}
\end{table}

We summarize the accuracy of TAPFixer and other approaches in Table \ref{otherwork1}. For these groups,  Liang et al. \cite{liang2016systematically}, MenShen \cite{bu2018systematically}, and AutoTap \cite{zhang2019autotap} fail or generate faulty patches. Liang et al. only partially fix the vulnerability in \emph{Group 1} \revise{(\textit{i.e.}, set a condition ``the user is at home'' in the first rule)} since it only focuses on current trigger conditions, but ignores the one derived from existing actions or new rules. For instance, the fix for the vulnerability in \emph{Group 1} and \emph{N/A 1} requires creating new rules, while the fix for \emph{Group 4} requires the action supplement. Liang et al. consider such fixes to be useless and are therefore classified as false negatives. For groups where they generate false positives (\emph{Group 2-3}), they remove the delay or change the temperature condition expected by the user, which violates the user's intent.
MenShen fails to fix all vulnerabilities since it lacks the formulation of rule semantics.
AutoTap fixes vulnerabilities in \emph{Group 1, 5} and \emph{N/A 1-2} to vary degrees with false positives detected in \emph{Group 1-4}.
The reason is that it merely focuses on the solution space derived from new rules but neglects the one associated with existing rules.
Thus, AutoTap does not prevent the expanded vulnerability from occurring but only reacts to it after it has occurred.
For example, AutoTap will \revise{generate} a patch for \emph{Group 4} to turn on the fan immediately after it is interrupted and still leave the fan to be interrupted when the \ce{CO2} level is high.

In comparison with existing approaches, TAPFixer can fix expanded vulnerabilities in all groups, including configuring from scratch. It is attributed to the accurate patch reasoned by our proposed negated-property reasoning algorithm. We have also tested the effectiveness of these fixes by deploying these TAP rules on HomeAssistant. In summary, TAPFixer achieves 100\% accuracy for this benchmark.
\vspace{-2mm}
\subsection{Market App Study of Violation Repair}
\label{marketapps}
To evaluate the generality of TAPFixer to repair market apps, we mainly focus on IFTTT and SmartThings.
Based on 9 modeled physical channels, we build a set of sensor groups that detect changes in the environment, including (1) temperature sensor; (2) light sensor; (3) presence sensor; (4) smoke sensor; (5) humidity sensor; (6) carbon monoxide sensor; (7) carbon dioxide sensor; (8) sound sensor; (9) weather sensor. By summarizing the general functional scenario of the actuator based on safety properties, we build a set of actuator groups, including (1) light; (2) door control, garage door control, and security camera; (3) air conditioner, heater, and electric blanket; (4) fan, window, sprinkler system, water valve, gas valve, natural gas hot water heater, and alarm; (5) smart plug, oven, and coffee maker.
We then randomly select 15-30 rules from the total rule set to form a rule group, which is associated with four categories of sensor groups and two categories of actuator groups at least. In total, we obtained 110 rule groups with an average of 21 rules per group.

\begin{table}[t!]
\centering
\caption{Summary of detection and repair results for G1-G7 with 110 rule groups, each of which contains 15-30 TAP rules.}
\resizebox{\linewidth}{!}
{
\begin{tabular}{|C{3cm}|C{1.5cm}|C{1.5cm}|C{1cm}|C{1.5cm}|c|}
\hline
{Application scenarios} & {Fixed violations} & {Unfixable violations} & {Safe cases} & {Generated patches} & {RSR$\uparrow$} \\\hline
\textbf{G1} (2 properties)         & 179                 & 35                      & 6            & 364                   & 83.64\%          \\\hline

\textbf{G2}
~(21 properties)      & 1675                & 277                     & 358          & 2228                  & 85.81\%          \\\hline
\textbf{G3}
~(6 properties)       & 459                 & 201                     & 0            & 902                   & 69.55\%          \\\hline
\textbf{G4}
~(8 properties)       & 687                 & 59                      & 134          & 1145                  & 92.09\%          \\\hline
\textbf{G5}
~(9 properties)       & 870                 & 68                      & 52           & 1288                  & 92.75\%          \\\hline
\textbf{G6}
~(3 properties)       & 272                 & 58                      & 0            & 491                   & 82.42\%          \\\hline
\textbf{G7}
~(4 properties)       & 402                 & 2                       & 36           & 419                   & 99.50\%          \\\hline
Total                     & 4544                & 700                     & 586          & 6837                  & 86.65\%     \\\hline
\end{tabular}
}
\label{appscenarios}
\vspace{-2mm}
\end{table}

\begin{figure}[t!]
\centering
\includegraphics[width=0.47\textwidth]{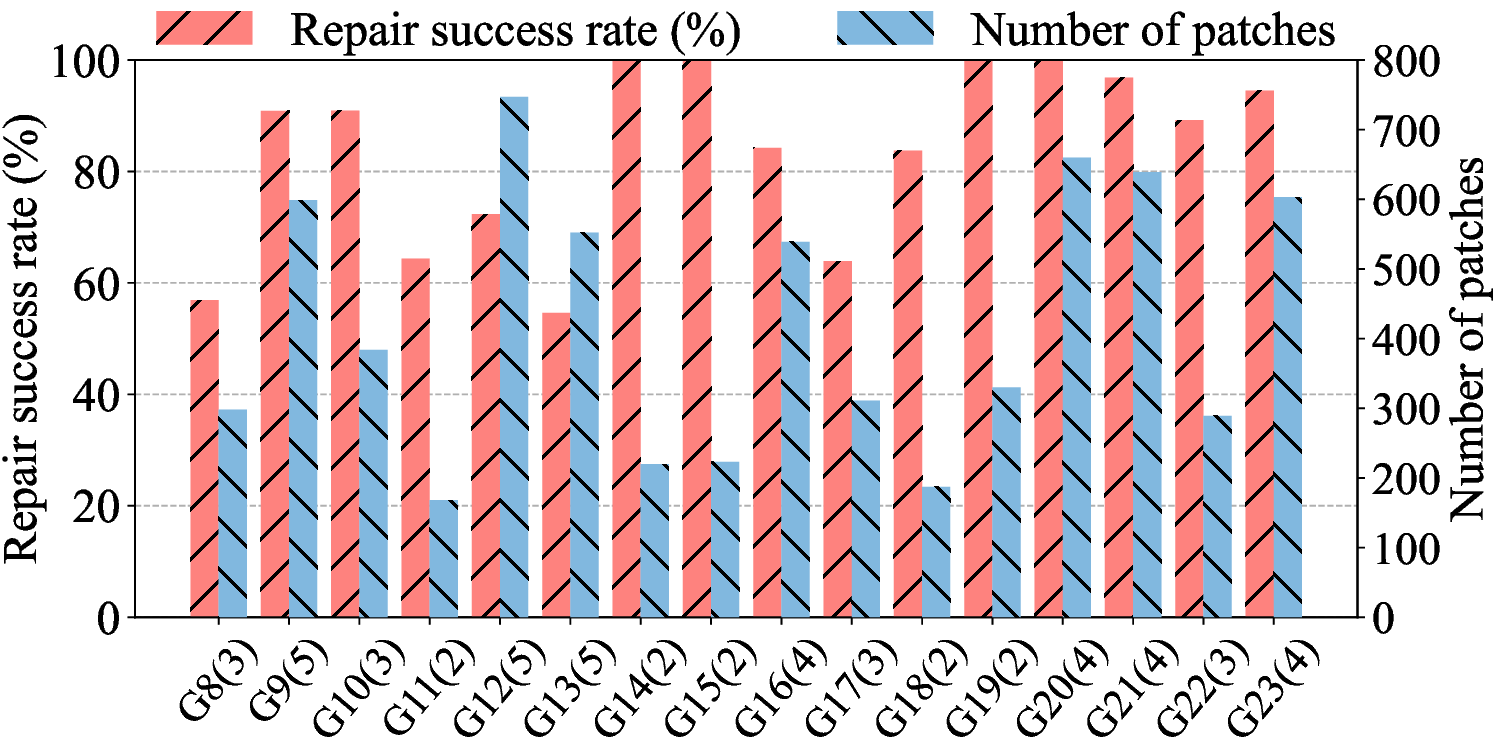}
\caption{Repair success rate and the number of generated patches for G8-G23 with 110 rule groups. The numbers in parentheses represent the number of properties contained in the property set.}
\label{fig:pic2}
\end{figure}

With these generated rule groups, we verify and fix rule interaction vulnerabilities based on correctness properties from two aspects: safety and prioritization of application scenarios. The former ensures that all modeling devices operate in a safe status (e.g., all appliances are turned off) in the same application scenarios (e.g., when no one is home), while the latter ensures that the same actuator can process safety requests with higher priority if there are conflicting safety properties in different application scenarios.
We categorize properties into 7 groups \textbf{(G1-G7)} by determining whether the application scenario is the same based on the pre-proposition of properties. We also categorize them into 16 groups \textbf{(G8-G23)} based on the prioritized status of the same actuator in the post-proposition of properties. The specific categories \textbf{(G1-G23)} for each property are shown in Appendix Table \ref{tab:sp}.

Following existing evaluation methods \cite{alhanahnah2020scalable, Celik_McDaniel_Tan_2018, celik2019iotguard, ding2021iotsafe}, we evaluate the accuracy of TAPFixer by manually checking found property violations and patches. The manual checking process is straightforward to perform given the output of TAPFixer since the number of rules is relatively small. The results of \textbf{G1-G7} are shown in Table \ref{appscenarios}. TAPFixer correctly fixes 4544 out of 5244 property violations, achieving an RSR 86.65\%. The results of \textbf{G8-G23} are shown in Fig. \ref{fig:pic2}. TAPFixer correctly fixes 4460 out of 5335 property violations, achieving an RSR of 83.60\%. For all properties in \textbf{G1-G23}, TAPFixer successfully repairs 9004 out of 10579 property violations, achieving a total RSR of 85.11\%.

TAPFixer fails to fix 1575 property violations in total. We manually analyze these failures and find that there are three major reasons:
(1) \emph{Predicate solving errors due to SMT}. Some local vulnerabilities can be fixed only by modifying some of all predicates. However, the SMT solver in nuXmv may unexpectedly specify the satisfiability of predicates that are irrelevant for violation repair, which can cause TAPFixer to fail to fix vulnerabilities correctly; (2) \emph{Limitations on the semantic abstraction of new transitions}. In model semantic abstraction, we currently only consider interpolating numerical variables as trigger predicates, not status predicates. It can invalidate scenarios that require numerical status predicates to fix vulnerabilities;
(3) \emph{Too large number of predicates}. We have limited the number of predicates to avoid state explosion. However, the number of predicates in some scenarios can still exceed the upper limit and cause TAPFixer to fail.

We further conduct a comparison with the SOTA AutoTAP \cite{zhang2019autotap}. Since AutoTAP cannot automatically extract rules from IoT apps, we randomly select 25 rule groups and manually write them using AutoTAP's interface to run AutoTAP.
AutoTAP only supports modeling limited rule features and does not support modeling of many device capabilities (e.g., garage door, heater, fan, sprinkler, etc.) and physical attributes (e.g., CO, \ce{CO2}, and humidity, etc), which may result in a low RSR due to modeling failures. Hence, to guarantee fairness, we filter out these correctness properties containing device capabilities (e.g., fan and sprinkler in \textbf{G6}) and physical attributes (e.g., CO and \ce{CO2} in \textbf{G5}) that are not supported by AutoTap, and finally obtain 22 correctness properties.
Besides, we introduce a metric \textit{Modeling Success Rate} (MSR) to assess the integrity of rule modeling and categorize \textit{Repair Failure Rate} (RFR) into RFR-MF and RFR-LIMIT to evaluate reasons for repair failures caused by modeling failures and repair algorithm limitations, respectively.

\begin{table}[t!]
\caption{Comparison between AutoTap and TAPFixer.}
\label{tab:RSR_compare}
\centering
\scalebox{0.7}
{
\begin{tabular}{|c|c|c|}
\hline
 Evaluation target & AutoTap & TAPFixer  \\ \hline
 MSR$\uparrow$       &      54.23\%            & 100\%    \\ \hline
RSR$\uparrow$ of \textbf{G1} (1 property)        &      20\%               & 44\%     \\ \hline
RSR$\uparrow$ of \textbf{G2} (14 properties)     &      51.43\%                   & 74.59\%     \\ \hline
RSR$\uparrow$ of \textbf{G3} (1 property)        &      8\%                & 92\%       \\ \hline
RSR$\uparrow$ of \textbf{G4} (4 properties)      &     44\%              & 94.32\%   \\ \hline
RSR$\uparrow$ of \textbf{G7} (2 properties)      &      38\%            & 91.49\%    \\ \hline
RFR-MF/RFR-LIMIT  $\downarrow$   &      23.99\%/24.57\%     & 0\%/20.93\%   \\ \hline
\end{tabular}
}
\vspace{-4mm}
\end{table}

We check MSR, RSR, RFR-MF, and RFR-LIMIT achieved by AutoTAP and TAPFixer and show the comparison results in Table \ref{tab:RSR_compare}.
In general, a lower MSR indicates that the constructed model contains fewer rule interactions and thus, is more likely to be reported as a successful fixed case and the RSR should be higher. However, the MSR and RSR of AutoTAP are both lower than TAPFixer. We can find that AutoTAP and TAPFixer totally achieve an RFR of 54.55\% and 20.93\%, respectively. Among these repair failures, after eliminating these failures caused by modeling failures, altough AutoTAP can model rules much less than TAPFixer, it still has a higher RFR-LIMIT (24.57\%).
The fundamental reason is that the repair capability of AutoTAP is lower than TAPFixer and may repair the rule group incorrectly. For example, in G2 with several rule groups violating the properties P.2, if the target light is controlled by a rule involving a numeric attribute (e.g., temperature), AutoTAP will not generate any repair patches. This is because AutoTAP directly filters out those actions activated by numeric states that could lead to other property violations, thereby avoiding additional analysis costs.
If not, AutoTAP can generate a patch rule (``If the light is off while the user is at home, then turn the light on''), which, however, cannot fix the violation.
Our TAPFixer can generate a correct patch rule ``If the user is not at home, then turn the light off'' both in these two cases since it uses an iterative abstraction and refinement process to correct and verify the feasibility of all generated possible patches at a low analysis cost.

\begin{figure}[t!]
\centering
\includegraphics[width=0.95\columnwidth]{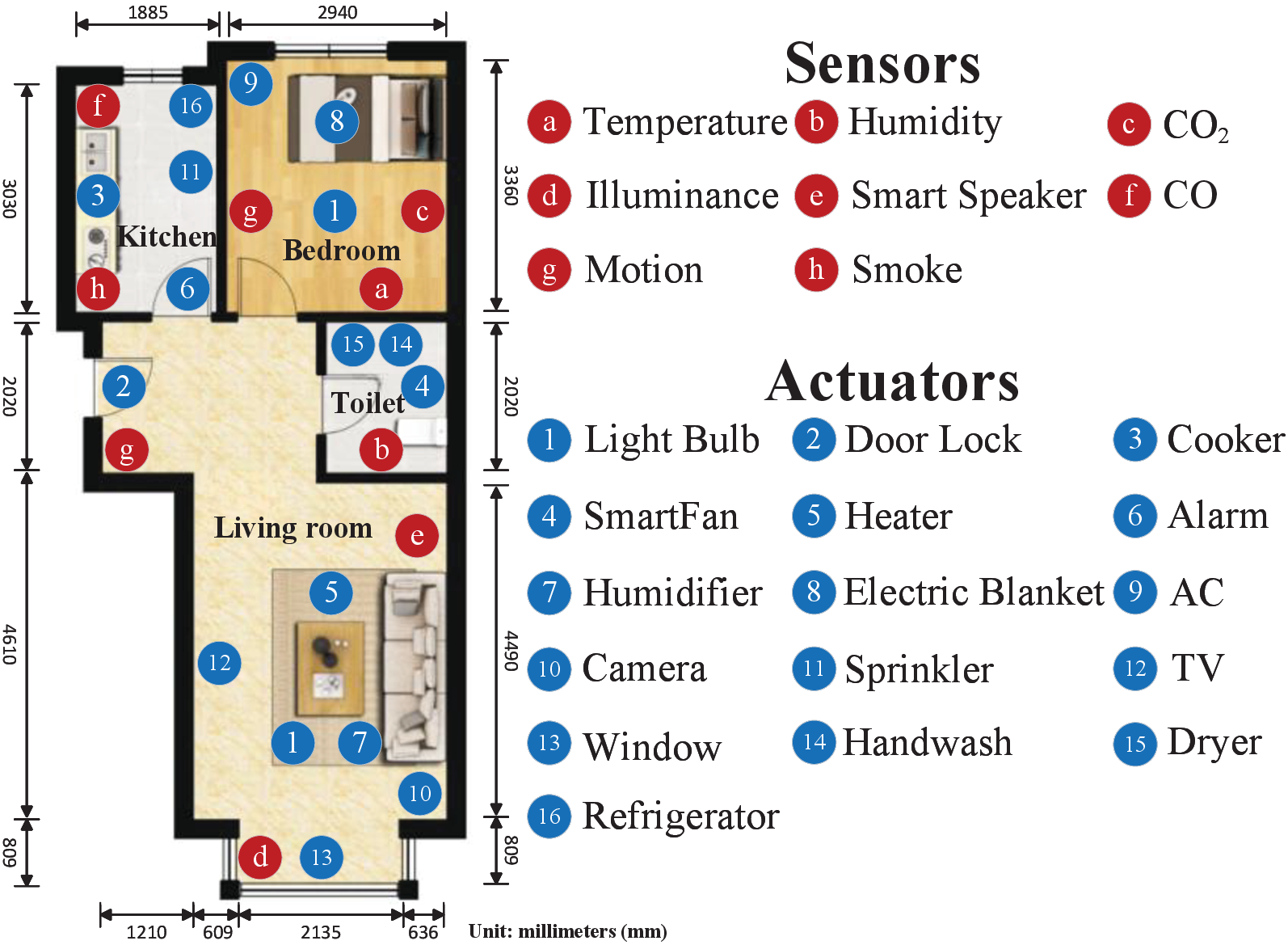}
\vspace{1mm}
\caption{Device layout in the smart home.}
\label{fig:house}
\end{figure}

\begin{table}[t!]
\centering
\caption{Number of identified and fixed vulnerabilities in 129 rules.}
\label{usertable}
\resizebox{\linewidth}{!}
{
\begin{tabular}{|c|c|c|c|c|c|c|c|c|}
\hline
                                                                          & V1   & V2  & V3     & V4    & V5  & V6    & V7    & V8     \\
\hline
\# found violations & 5    & 6   & 9      & 23    & 0   & 4     & 5     & 3      \\
\hline
\# fixed violations & 3    & 6   & 8      & 23    & 0   & 4     & 5     & 3      \\
\hline
RSR                                                              & 60\% & 100\% & 89.9\% & 100\% & N/A & 100\% & 100\% & 100\%  \\
\hline
\end{tabular}
}
\end{table}

\subsection{User Study}
\revise{To evaluate users' acceptance of patches generated by TAPFixer, we conduct a user study on the online questionnaire platform. We use the conceptual HA scenario in Fig. \ref{fig:house} and invite participants to set up TAP rules in the scenario. Since the scenario is fixed and our predefined properties in Table \ref{tab:sp} can almost cover all possible cases, we also require participants to pick up properties in Table \ref{tab:sp} for violation detection and repair.} We received 23 questionnaires and the distributions of our participants are as follows: 47.8\% are male and 52.2\% are female, 87.0\% are 20-35 years old, and 13\% are over 35 years old, 52.2\% indicate that they have professional experience in computer science, 78.3\% own smart devices and $73.9\%$ are familiar with home automation rules.

\revise{Given questionnaires specified in natural language, we manually unify different descriptions having the same meanings into the same term, so that we can map entities in users’ rule descriptions into TAP rule syntax accurately and extract TAP rules from these descriptions.} We obtain 394 TAP rules with an average of 17 rules per questionnaire taking 28 minutes. We then run TAPFixer to detect and repair interaction vulnerabilities in collected rules. We randomly select found property violations and manually analyze their vulnerability patterns that cause property violations as shown in Table \ref{usertable}. Overall, we analyze 129 flawed rules and TAPFixer achieves an RSR of $94.5\%$ in both basic and expanded vulnerability patterns.

Our findings from this study are as follows:
(1) the number of V4 is noticeably more than other identified vulnerabilities. Participants tend to set switches for temperature-related actuators to two separate TAP rules. This results in a vulnerability of V4 among the window, AC, heater, and electric blanket;
(2) although the majority of participants have used or learned about home automation, they rarely use conditions of TAP rules and set the delay for device executions, which affect the incidence of condition interference and latency-related issues to some extent;
(3) no V5 vulnerability is found since it mainly originates from platform delays and users' careless misconfiguration, which do not occur in the study.

\revise{To evaluate the quality of TAPFixer’s results, we randomly selected 9 property violations with their corresponding TAP rules, original rule descriptions, and generated rule patches to construct a questionnaire to feedback to 23 participants. The average completion time of this questionnaire is 7.3 minutes. We received a total of 103 strong agreements, 92 agreements, and 12 disagreements (3 in V1, 1 in V2, 2 in V3, 1 in V4, 1 in V7, and 4 in V8). The reasons for disagreements can be classified into two types: 1) negating violations; 2) proposing other patches, some of which we found are not able to fix the violation or are similar to patches. Hence, we further explained the results in detail to participants who disagreed and received 10 agreements again. So, TAPFixer achieved a 99.0\% satisfaction rate. The remaining disagreeing participants suggested adding restrictions on rules' lifespans in the patch (e.g., restricting $r_4$ in Fig.\ref{fig:rin}(b) to only run in the summer), which is feasible but not supported by TAPFixer yet.}

\vspace{-2mm}
\subsection{Performance Analysis}

\revise{Since the complexity of TAPFixer depends on input models, properties, and abstraction and refinement strategies, it is hard to directly quantify. Hence,} we evaluate the performance of the overall process in TAPFixer, including rule interaction modeling, vulnerability detection, and repair patch generation. To early terminate the oversized predicate exploration, we set the $ROUND\_LIMIT$ and $ITER\_LIMIT$ in NPR to 15 and 50, respectively. Within these two iterations, we record the execution time from two perspectives: the benchmark dataset in $\S$\ref{casestudy} and the realistic market dataset in $\S$\ref{marketapps}.

We first record TAPFixer's analysis time of each test case. To form the 21-rule test case, for each benchmark dataset (\emph{Group 1-5}), we combine rules presented in Table \ref{flawed} with randomly selected ones from the extracted rule set. TAPFixer does not complement initialization scenarios without rules (\emph{N/A 1-2}). Fig. \ref{fig:pic1} shows the time for verifying and repairing vulnerabilities in all 7 test cases. For 21-rule test cases (\emph{Group 1-5}), the longest, shortest, average time takes 267.3 seconds in \emph{Group 5}, 93.8 seconds in \emph{Group 3}, and 161.7 seconds respectively. The time of the initialization scenario \emph{N/A 1} and \emph{N/A 2} takes 241 and 215 ms respectively.

We then record TAPFixer's total time of evaluations on market apps in Table \ref{table:average}. Although the verification and repair time of \emph{G1-G7} takes a total of one hour longer than that of \emph{G8-G23} which takes around 6 hours, their average time to create a property violation patch takes about 3 seconds with an overall average of 3.51 seconds. According to the investigation on the efficiency of manual repair conducted by a previous study\cite{liang2016systematically}, many users give up debugging a property violation after 45 seconds. Hence, TAPFixer achieves over 15 times improvement compared to manually repairing a property violation.

\subsection{Discussion and Limitations}
\label{subsec:discussion}

\textbf{Scalability and Extension}. The proposed techniques in NPR are not limited to any specific platform or scenario in HA systems. They can be applied in other cases where TAP rules are widely used, such as UAV control, hardware description rules, industrial equipment control, etc.
In addition to these studied safety properties (\textit{i.e.}, something bad should never happen), we aim to introduce liveness properties (\textit{i.e.}, something good eventually happens) to address latency-related issues in more detail in the future.

\begin{figure}[t!]
\centering
\includegraphics[width=0.9\columnwidth]{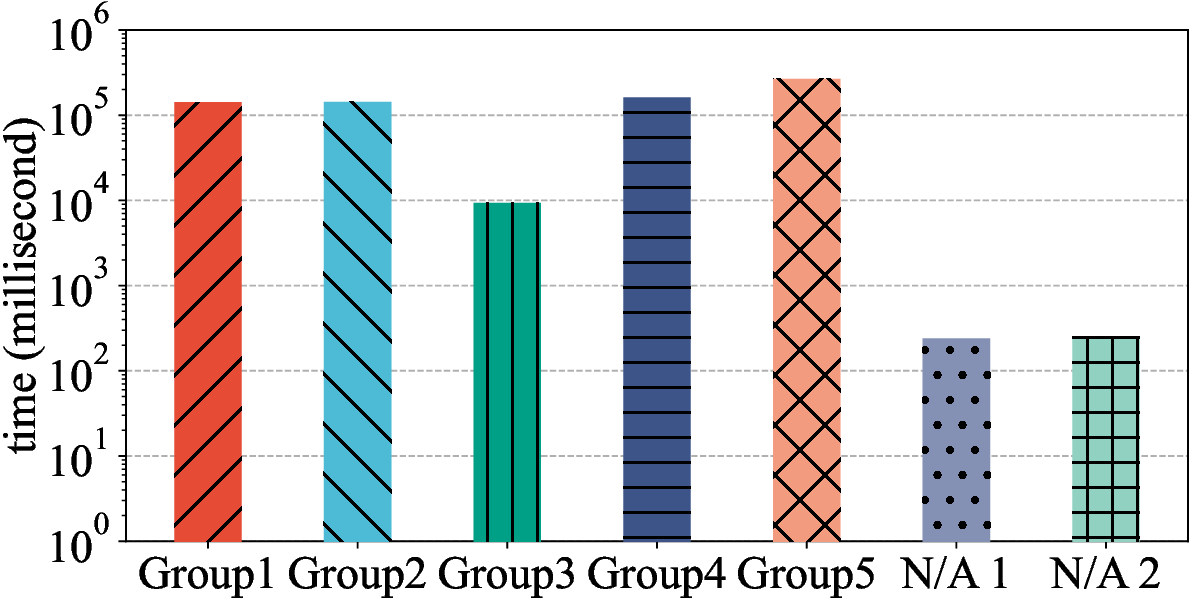}
\vspace{1mm}
\caption{Verification and repair time of each 21-rule benchmark dataset and initialization scenarios.}
\label{fig:pic1}
\vspace{-4mm}
\end{figure}

\definecolor{MineShaft}{rgb}{0.2,0.2,0.2}
\begin{table} [t!]
\centering
\caption{Average patch generation time for market apps.}
\label{table:average}
\resizebox{\linewidth}{!}
{
\begin{tabular}{|C{4cm}|C{1.5cm}|C{2.5cm}|C{3.5cm}|}
\hline
   Market apps  &  Total time (minute) & Number of generated patches & Avg. generation time per patch (second)\\\hline
   \textbf{G1-G7} (110 rule groups)                                      & 364.070              & 6837                  & 3.195                                   \\\hline
    \textbf{G8-G23} (110 rule groups)                                     & 431.261              & 6749                  & 3.834 \\\hline
\end{tabular}
}
\end{table}

\textbf{Limitation}. TAPFixer achieves high effectiveness and speed in verifying and repairing vulnerabilities in complex logical and physical spaces, but it still suffers from three limitations: First, although TAPFixer can achieve an RSR of $86.65\%$ overall, three vulnerability cases cannot be fixed: (1) to address predicate-solving errors, we find that such failures can be avoided by running NPR again or setting the abstraction target to a single rule instead of the rule collection; (2) to address limitations during the semantics abstraction of new rules ($\S$\ref{rule semantics abstract}), future work can extend TAPFixer with numerical condition abstraction to address this limitation; (3) possible solutions to avoid the predicate scale timeout are to increase $ITER\_LIMIT$, $ROUND\_LIMIT$ and further optimize the predicate scale. Second, TAPFixer relies on our manual qualitative analysis of device effects on physical channels to enhance the comprehensiveness of its modeling capabilities. However, the accuracy of effect models is hard to guarantee as the home environment changes, which limits the performance of TAPFixer. But, fortunately, the home environment will not change frequently once devices are deployed in. In the future, we can introduce an online miner to regularly update effect models, which collects sensor data over a long period and uses SVM to mine the qualitative relationships between channels and device activities, similar to \cite{birnbach2019peeves}. Third, TAPFixer currently can only provide repair patches. In the future, more reverse engineering techniques should be exploited to generate program patches of rule patches confirmed by users for specific HA platforms. Take SmartThings for example, we can first map the patch into the corresponding inter-procedural control flow graph extracted from apps during rule modeling (see Appendix \ref{asec:mie}) and then transform it into the form of the Groovy abstract syntax tree which can be easily transformed into Groovy codes using Groovy AST Transformation \cite{asttran}.

\section{Related work}
\label{sec:rel}
With the prevalence of home automation, safety issues in the TAP-based HA system have gained much attention. \revise{Table \ref{tab:detect-work} and Table \ref{tab:compare} respectively illustrate the comparison of our research with related detection and repair work over various features}. We discuss this comparison as follows:

\begin{table}[]
\caption{ \revise{Comparison between TAPFixer and related detection works.}}
\label{tab:detect-work}
\resizebox{\linewidth}{!}{
\begin{tabular}{|c|c|c|c|c|c|c|}
\hline
Related work & Latency  & Tardy attribute & Implicit effect & Joint physical effect & Nondeterminacy \\ \hline
IOTCOM\cite{Dillig_Dillig_Aiken_2008}       & \ding{52}  & \ding{56} & \ding{52} & \ding{56} & \ding{56} \\ \hline
SAFECHAIN\cite{hsu2019safechain}    & \ding{56}  & \ding{56} & \ding{56} & \ding{56} & \ding{56} \\ \hline
IRuler\cite{wang2019charting}      &  \ding{52}  & \ding{56}  & \ding{56} & \ding{56} & \ding{56} \\ \hline
IoTGUARD\cite{celik2019iotguard}     & \ding{56}  & \ding{56} & \ding{56} & \ding{56} & \ding{56} \\ \hline
SOATERIA\cite{Celik_McDaniel_Tan_2018}     & \ding{56}  & \ding{56} & \ding{56} & \ding{56} & \ding{56} \\ \hline
TAPInspector\cite{yu2022tapinspector} & \ding{52}  & \ding{52} & \ding{52} & \ding{56} & \ding{56} \\ \hline
IOTSAN\cite{nguyen2018iotsan}       &  \ding{56}  & \ding{56} & \ding{56} & \ding{56} & \ding{56} \\ \hline
HOMEGUARD\cite{chi2020cross}    &  \ding{52}  & \ding{56} & \ding{56} & \ding{56} & \ding{56} \\ \hline
Jia et al.\cite{jia2021s}   &  \ding{56}  & \ding{56} & \ding{56}& \ding{56} & \ding{56} \\ \hline
IoTSEER\cite{ozmen2022discovering}      &  \ding{52}  & \ding{52}  & \ding{52} & \ding{52} & \ding{56} \\ \hline
Chi et al.\cite{chi2022delay}   &  \ding{52}  & \ding{56} & \ding{56} & \ding{56} & \ding{56} \\ \hline
TAPFixer     & \ding{52}  & \ding{52} & \ding{52} & \ding{52} & \ding{52} \\ \hline
\end{tabular}
}
\vspace{-4mm}
\end{table}

\begin{table}[t!]
\caption{Comparison between TAPFixer and related repair works. $phy_1$, $phy_2$, $phy_3$, $phy_4$, and $phy_5$ denote tardy attribute, channel-based interaction, implicit physical effect, joint physical effect, nondeterminacy, respectively.}
\label{tab:compare}
\resizebox{\linewidth}{!}{
\begin{tabular}{|c|ccc|ccccc|}
\hline
\multirow{2}{*}{\begin{tabular}[c]{@{}c@{}}Related work\end{tabular}} &
  \multicolumn{3}{c|}{Latency-realted} &
  \multicolumn{5}{c|}{Physical-related} \\ \cline{2-9}
 &
  \multicolumn{1}{c|}{$l_1$} &
  \multicolumn{1}{c|}{$l_2$} &
  $l_3$ &
  \multicolumn{1}{c|}{$phy_1$} &
  \multicolumn{1}{c|}{$phy_2$} &
  \multicolumn{1}{c|}{$phy_3$} &
  \multicolumn{1}{c|}{$phy_4$} &
  $phy_5$ \\ \hline
ESOs\cite{schuster2018situational} &
  \multicolumn{1}{c|}{\ding{52}} &
  \multicolumn{1}{c|}{\ding{52}} & \ding{52}
   &
  \multicolumn{1}{c|}{\ding{52}} &
  \multicolumn{1}{c|}{\ding{52}} &
  \multicolumn{1}{c|}{\ding{56}} &
  \multicolumn{1}{c|}{\ding{56}} & \ding{52}
   \\ \hline
Bastys et al.\cite{bastys2018if} &
  \multicolumn{1}{c|}{\ding{52}} &
  \multicolumn{1}{c|}{\ding{52}} &\ding{52}
   &
  \multicolumn{1}{c|}{\ding{52}} &
  \multicolumn{1}{c|}{\ding{56}} &
  \multicolumn{1}{c|}{\ding{56}} &
  \multicolumn{1}{c|}{\ding{56}} & \ding{52}
   \\ \hline
He et al.\cite{he2018rethinking} &
  \multicolumn{1}{c|}{\ding{52}} &
  \multicolumn{1}{c|}{\ding{56}} &\ding{56}
   &
  \multicolumn{1}{c|}{\ding{52}} &
  \multicolumn{1}{c|}{\ding{56}} &
  \multicolumn{1}{c|}{\ding{56}} &
  \multicolumn{1}{c|}{\ding{56}} & \ding{52}
   \\ \hline
SmartAuth\cite{tian2017smartauth} &
  \multicolumn{1}{c|}{\ding{52}} &
  \multicolumn{1}{c|}{\ding{52}} &\ding{52}
   &
  \multicolumn{1}{c|}{\ding{52}} &
  \multicolumn{1}{c|}{\ding{52}} &
  \multicolumn{1}{c|}{\ding{56}} &
  \multicolumn{1}{c|}{\ding{56}} & \ding{52}
   \\ \hline
ContexIoT\cite{jia2017contexlot} &
  \multicolumn{1}{c|}{\ding{52}} &
  \multicolumn{1}{c|}{\ding{52}} & \ding{52}
   &
  \multicolumn{1}{c|}{\ding{52}} &
  \multicolumn{1}{c|}{\ding{52}} &
  \multicolumn{1}{c|}{\ding{56}} &
  \multicolumn{1}{c|}{\ding{56}} & \ding{52}
   \\ \hline
IoTSAFE\cite{ding2021iotsafe}&
  \multicolumn{1}{c|}{\ding{52}} &
  \multicolumn{1}{c|}{\ding{52}} &\ding{56}
   &
  \multicolumn{1}{c|}{\ding{52}} &
  \multicolumn{1}{c|}{\ding{52}} &
  \multicolumn{1}{c|}{\ding{52}} &
  \multicolumn{1}{c|}{\ding{52}} & \ding{52}
   \\ \hline
IoTMEDIATOR\cite{Chi_Zeng_Du}&
  \multicolumn{1}{c|}{\ding{52}} &
  \multicolumn{1}{c|}{\ding{56}} &\ding{56}
   &
  \multicolumn{1}{c|}{\ding{52}} &
  \multicolumn{1}{c|}{\ding{52}} &
  \multicolumn{1}{c|}{\ding{56}} &
  \multicolumn{1}{c|}{\ding{56}} & \ding{52}
   \\ \hline
Liang et al.\cite{liang2016systematically} &
  \multicolumn{1}{c|}{\ding{52}} &
  \multicolumn{1}{c|}{\ding{56}} &\ding{52}
   &
  \multicolumn{1}{c|}{\ding{52}} &
  \multicolumn{1}{c|}{\ding{56}} &
  \multicolumn{1}{c|}{\ding{56}} &
  \multicolumn{1}{c|}{\ding{56}} &
   \ding{56} \\ \hline
MenShen\cite{bu2018systematically} &
  \multicolumn{1}{c|}{\ding{52}} &
  \multicolumn{1}{c|}{\ding{56}} &\ding{52}
   &
  \multicolumn{1}{c|}{\ding{52}} &
  \multicolumn{1}{c|}{\ding{56}} &
  \multicolumn{1}{c|}{\ding{56}} &
  \multicolumn{1}{c|}{\ding{56}} &\ding{56}\\ \hline
AutoTap\cite{zhang2019autotap} &
  \multicolumn{1}{c|}{\ding{52}} &
  \multicolumn{1}{c|}{\ding{56}} &\ding{56}
   &
  \multicolumn{1}{c|}{\ding{56}} &
  \multicolumn{1}{c|}{\ding{56}} &
  \multicolumn{1}{c|}{\ding{56}} &
  \multicolumn{1}{c|}{\ding{56}} &\ding{56}
   \\ \hline
TAPFixer &
  \multicolumn{1}{c|}{\ding{52}} &
  \multicolumn{1}{c|}{\ding{52}} &
  \ding{52} &
  \multicolumn{1}{c|}{\ding{52}} &
  \multicolumn{1}{c|}{\ding{52}} &
  \multicolumn{1}{c|}{\ding{52}} &
  \multicolumn{1}{c|}{\ding{52}} &
  \ding{52} \\ \hline
\end{tabular}
}
\vspace{1mm}
\end{table}

\textbf{Rule Vulnerability Detection}. Detecting vulnerabilities in more comprehensive and detailed environments has been widely studied for the past years. SOTERIA\cite{Celik_McDaniel_Tan_2018} proposes a detection method based on LTL model checking. IOTSAN\cite{nguyen2018iotsan} analyses sequentiality and concurrency of rule executions. \revise{IOTA\cite{Newcomb_Chandra_Jeannin_Schlesinger_Sridharan_2017} proposes the calculus for the domain of home automation to secure rule interactions. Balliu et al.\cite{Balliu_Merro_Pasqua_2019} proposes a semantic framework capturing the essence of cross-app interactions and Friendly Fire\cite{balliu2021friendly} further optimises it. SAFECHAIN\cite{hsu2019safechain} detects hidden attack chains exploiting the combination of rules based on model checking.} IOTCOM\cite{alhanahnah2020scalable} focuses on vulnerabilities between logical and physical interactions. HOMEGUARD\cite{chi2020cross} uses SMT techniques to detect cross-app conflicts. IRuler\cite{wang2019charting} considers the uncertainty of smart devices. Jia et al.\cite{jia2021s} handle scenarios about device management channels. Chi et al.\cite{chi2022delay} studies vulnerabilities caused by delay-based automation interference attacks. TAPInspector\cite{yu2022tapinspector} implements a comprehensive analysis of rule interactions by introducing latency and connection-based features. IoTSEER\cite{ozmen2022discovering} alleviates approximation problems through dynamic analysis.
 \revise{TAPFixer follows these advanced methods to model and detect rules. But, we introduce more features to improve the accuracy and integrity of static modeling shown in Table \ref{tab:detect-work} so that NPR can statically generate formal-soundness patches using formal verification.}

\textbf{Dynamic Rule Enforcement Control}. Existing works develop control policies based on specific concerns in the HA system and dynamically prevent risks.
ContexIoT\cite{jia2017contexlot} uses the data and control flows of smart apps to build access contexts.
SmartAuth\cite{tian2017smartauth} investigates authorization mechanisms with different behavioral security levels.
Bastys et al.\cite{bastys2018if} design a short-term and a long-term access control mechanism based on the information flow.
He et al.\cite{he2018rethinking} point out the key factors that constitute the complex access control scenario.
ESOs\cite{schuster2018situational} studies cross-layer access control and provides corresponding permissions.
IoTSAFE\cite{ding2021iotsafe} identifies real-time physical interactions combined with dynamic and static methods to predict and avoid hazard scenarios. IoTMEDIATOR\cite{Chi_Zeng_Du} provides a more fine-grained access control framework through threat-tailored handling. \revise{The comparison between TAPFixer and advanced dynamic-based repair works is shown in Table \ref{tab:compare}.} Compared to these methods, TAPFixer can eliminate the root cause of vulnerabilities.

 \textbf{Static Fixing Towards Rule Semantic}.
Researches on static fixing include Liang et al.\cite{liang2016systematically}, MenShen\cite{bu2018systematically}, AutoTap\cite{zhang2019autotap}, and also our work TAPFixer. AutoTap\cite{zhang2019autotap} proposes an automaton-based automatic method to fix vulnerabilities, which identifies the violated bridge edge and generates fixes based on it. Liang et al.\cite{liang2016systematically} and MenShen\cite{bu2018systematically} develop semi-automatic formal methods to fix vulnerabilities. They parameterize the syntax of existing rules and solve for specific values that can eliminate the violation.
However, none of them perform detailed physical and latency analysis, missing many practical features and vulnerabilities. Additionally,
their fixing algorithms are limited in addressing expanded vulnerabilities associated with complex physical and latency issues. \revise{The comparison between TAPFixer and advanced static-based repair works is also shown in Table \ref{tab:compare}.} Compared to these methods, TAPFixer can effectively repair vulnerabilities in complex physical environments, especially when dealing with expanded vulnerabilities.

\vspace{-2mm}
\section{Conclusion}
In this work, we design a novel automatic vulnerability detection and repair framework, TAPFixer, for TAP-based home automation systems. With a comprehensive analysis of existing rule interaction vulnerabilities, TAPFixer can model TAP rules with more practical latency and physical features to capture the accurate rule execution behaviors both in the logical and physical space and identify more interaction vulnerabilities. We propose a novel negated-property reasoning algorithm for TAPFixer so that it is able to accurately generate valid patches for eliminating vulnerabilities both in the logical and physical space. We conduct numerical evaluations of TAPFixer from aspects of accuracy analysis, repair capabilities of market apps, real user study, and execution performance. The results of our evaluation show that TAPFixer can correctly fix different rule interaction vulnerabilities with an excellent performance overhead.

\section*{Acknowledgment}
We thank our shepherd and the anonymous reviewers for their valuable feedback. This work was partially supported by the National Natural Science Foundation of China under Grant 62202387 and the Guangdong Basic and Applied Basic Research Foundation under Grant 2021A1515110279.

\bibliographystyle{plain}
\bibliography{main}

\appendix
\section{Rule Interaction Vulnerability Patterns}
\label{asec:rivp}
TAPFixer integrates more features in physical space and uncovers 8 patterns of rule interaction vulnerability which are summarized in Fig. \ref{fig:vul}.

\textbf{V1:}  \textit{Trigger-Interference Basic Pattern.} $a_i$ of $r_i$ and $t_j$ of $r_j$ share the same immediate channel attribute ($a_{i} \cap t_{j} \cap \mathbbm{A}_{immd} \ne \emptyset$), which can cause events generated by $a_i$ unexpectedly triggering $r_j$ ($a_{i}\dashrightarrow t_{j}$) and put the rule execution at risks.

\textbf{V2:} \textit{Condition-Interference Basic Pattern.}
$a_i$ of $r_i$ and $c_j$ of $r_j$ share the same immediate channel attribute ($a_{i} \cap c_{j} \cap \mathbbm{A}_{immd} \ne \emptyset$), which may lead to $a_i$ changing the satisfaction of $c_j$ ($a_{i}\Rightarrow c_{j}$) and change the rule context defectively.

\textbf{V3:} \textit{Action-Interference Basic Pattern.} $r_i$ and $r_j$ are both immediate rules ($r_{i}, r_{j} \in r_{immd}$) and have no latency ($ l_{1} \notin r_{i}\cup r_{j}$). Commands from $a_i$ and $a_j$ to the same device are conflicting ($a_{i} \xrightarrow{\times} a_{j}$) and may override the effect after the secure interaction.

\textbf{V4:} \textit{Tardy-channel-based Trigger
Interference.} Similar to V1, $a_i$ can unexpectedly trigger $t_j$ sharing the same physical channel ($a_{i}\dashrightarrow t_{j}$), but from the tardy channel ($a_{i} \cap t_{j} \cap \mathbbm{A}_{tardy} \ne\emptyset$). \textit{Group1} in Table \ref{tab:frg} shows an example of V4.

\textbf{V5:} \textit{Disordered
Action Scheduling}.
$r_i$ and $r_j$ are both immediate ($r_{i}, r_{j} \in r_{immd}$) and have conflicting actions ($a_{i} \xrightarrow{\times} a_{j}$), but include $l_1$ compared to V3 ($l_{1} \in r_{i}\cup r_{j}$).
$r_i$ and $r_j$ will be triggered simultaneously ($t_{i} \cap t_{j} \ne\emptyset$),
but with the involvement of $l_i$ and $l_j$ (if it exists), the expected order of actions is disrupted. \textit{Group2} in Table \ref{tab:frg} shows an example of V5.

\textbf{V6:} \textit{Action Overriding.}
This type of vulnerability is similar to V5, but $r_i$ and $r_j$ can be triggered separately ($t_{i} \cap t_{j} =\emptyset$). The execution time $l_i$ is longer than $l_j$ (if it exists), causing $a_i$ to overwrite the effect of the previous execution of $a_j$ ($a_{i} \xrightarrow{\times} a_{j}$).
\textit{Group3} in Table \ref{tab:frg} shows an example of V6.

\begin{figure}[t]
\begin{center}
\setlength{\abovecaptionskip}{0cm}
\setlength{\belowcaptionskip}{0cm}
\begin{minipage}{1\columnwidth}
  \centering
    \subfloat[\textbf{V1}: Trigger-Interference\\ Basic Pattern.]{\includegraphics[width=0.49\columnwidth]{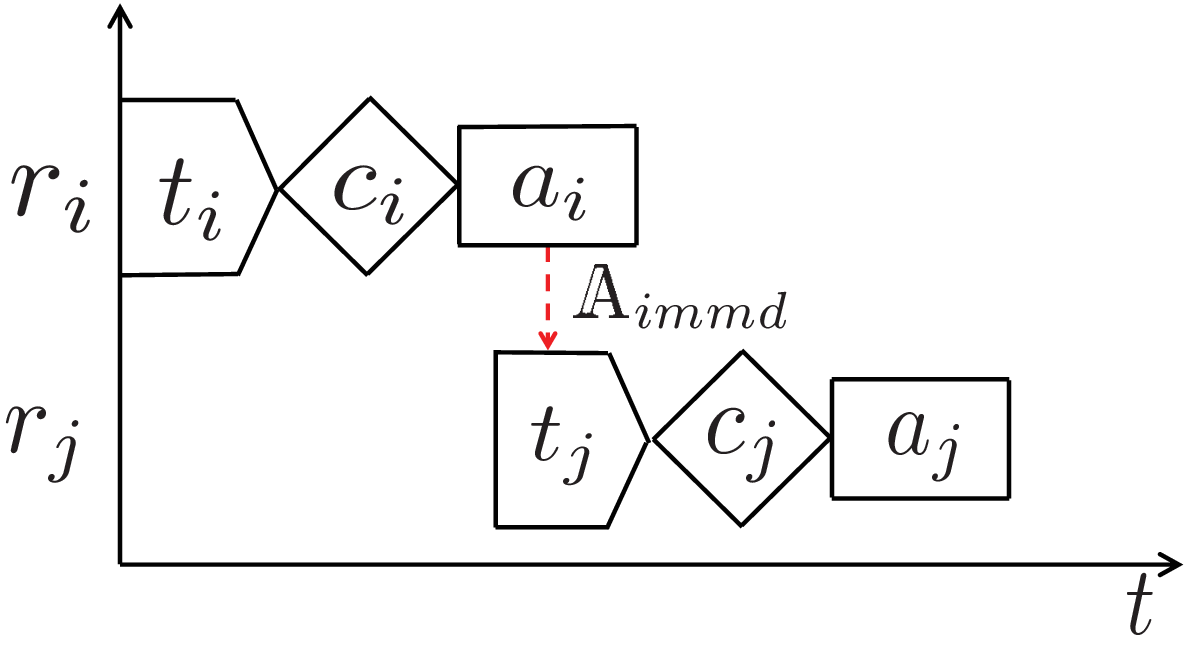}}
    \subfloat[\textbf{V2}: Condition-Interference\\ Basic Pattern.]{\includegraphics[width=0.49\columnwidth]{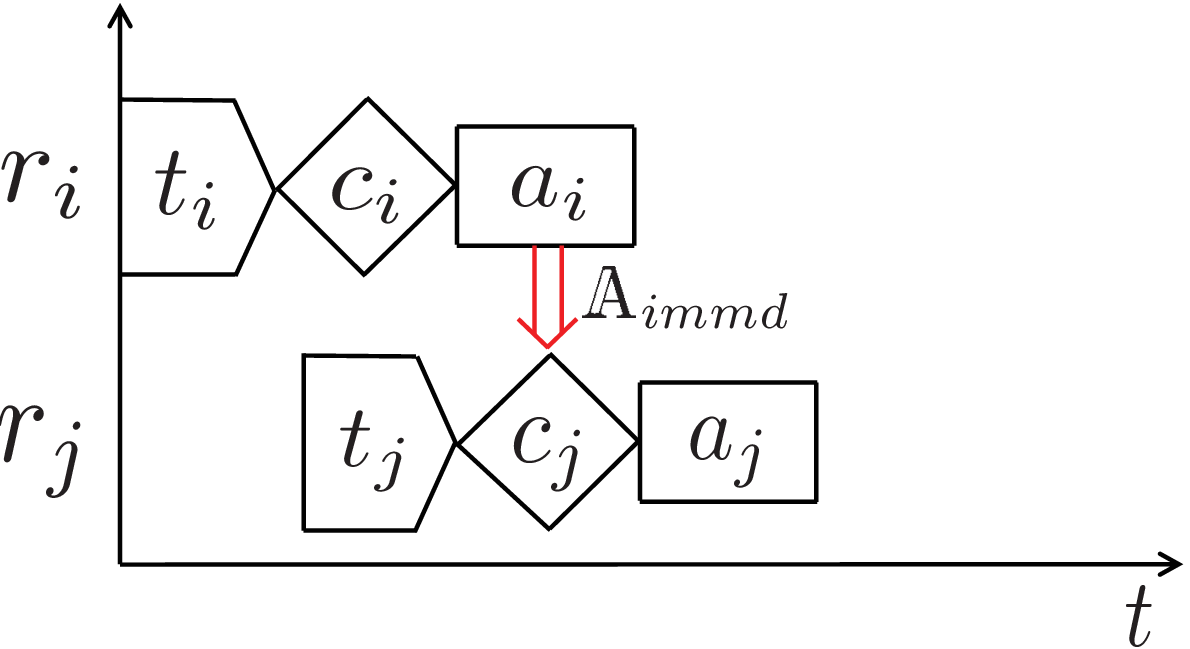}}
\end{minipage}

\begin{minipage}{1\columnwidth}
  \centering
    \subfloat[\textbf{V3}: Action-Interference\\ Basic Pattern.]{\includegraphics[width=0.49\columnwidth]{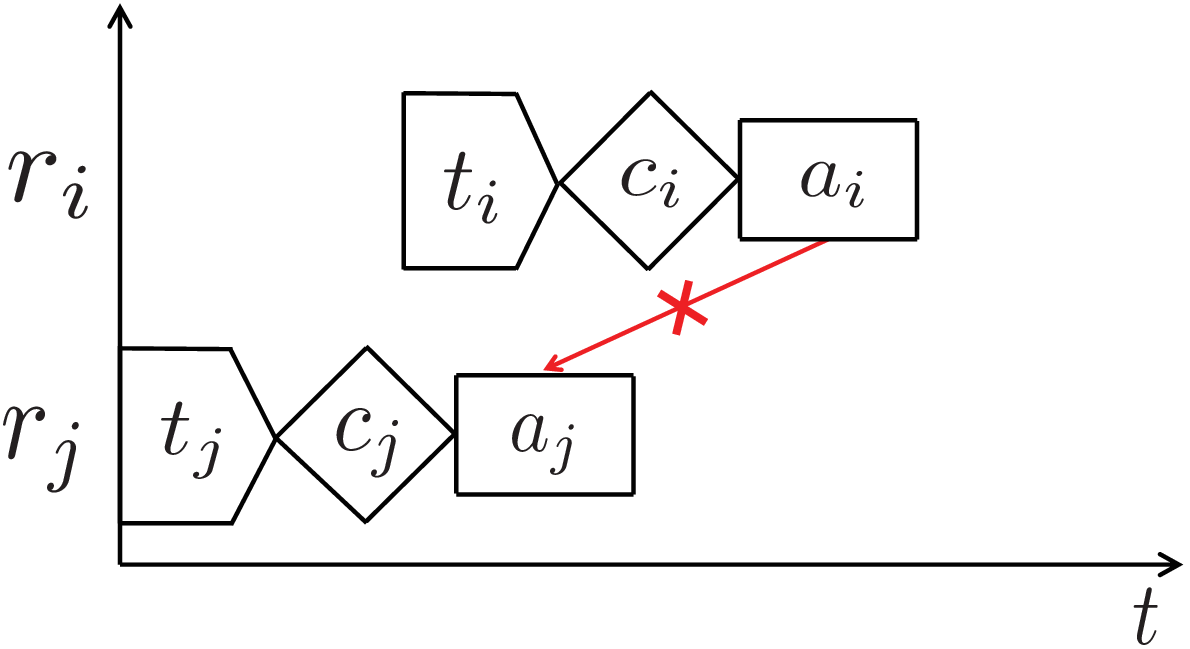}}
    \subfloat[\textbf{V4}: Tardy-channel-based\\ Trigger Interference.]{\includegraphics[width=0.49\columnwidth]{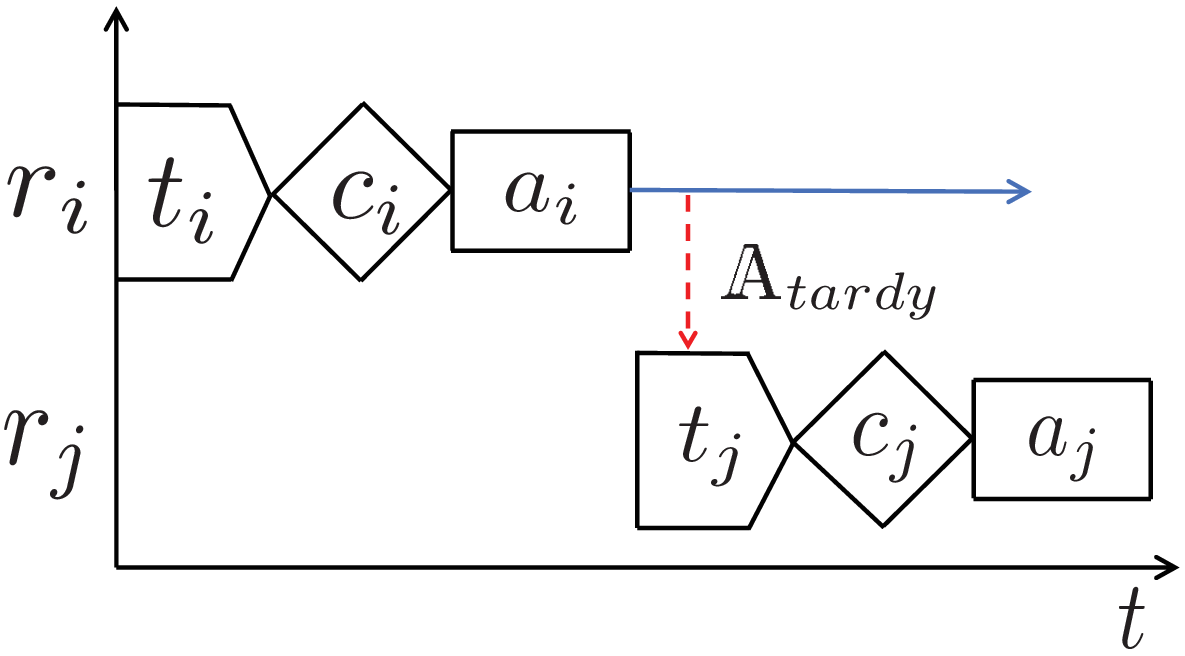}}
\end{minipage}

\begin{minipage}{1\columnwidth}
  \centering
    \subfloat[\textbf{V5}: Disordered Action\\ Scheduling.]{\includegraphics[width=0.49\columnwidth]{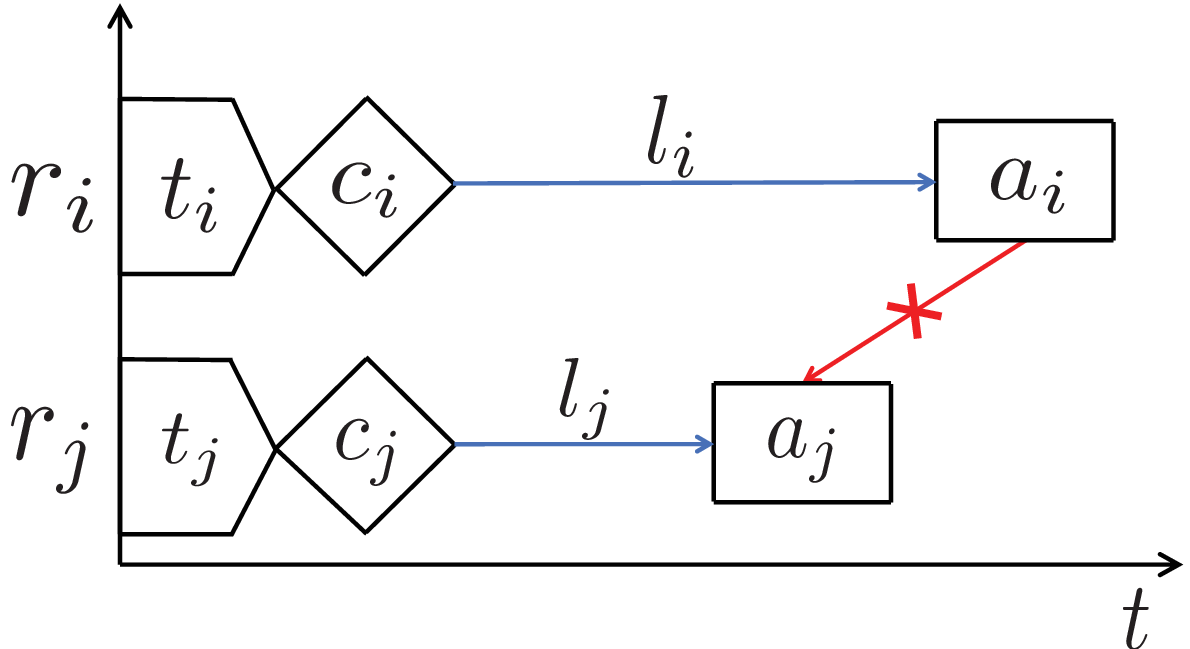}}
    \subfloat[\textbf{V6}: Action Overriding.]{\includegraphics[width=0.49\columnwidth]{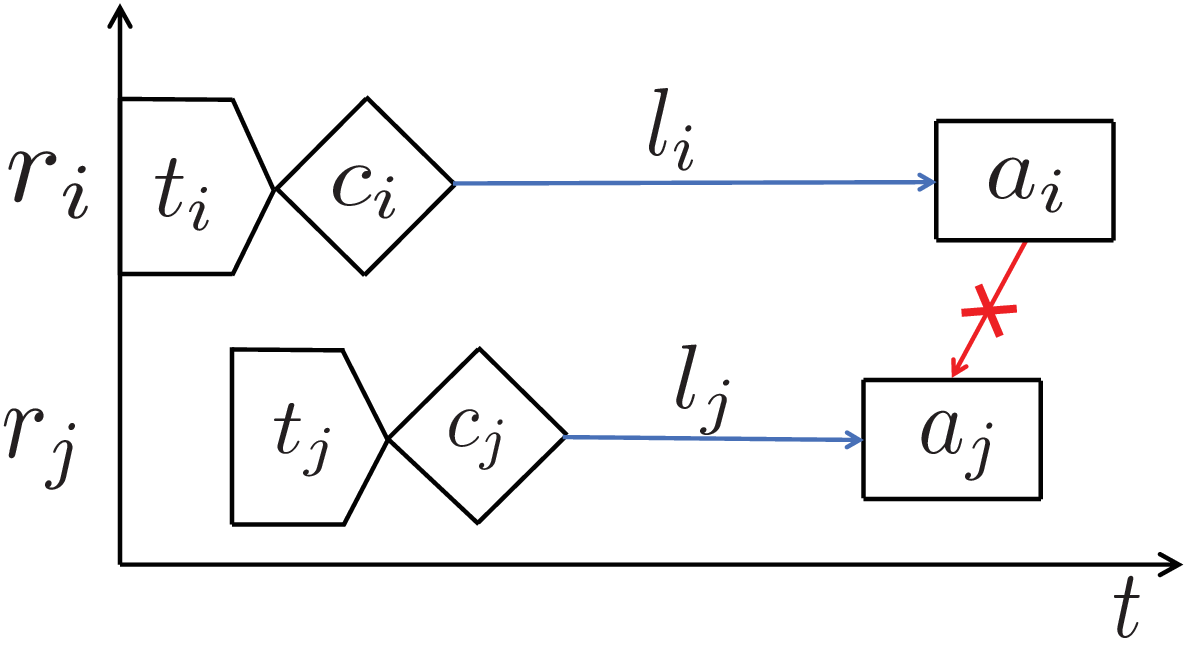}}
\end{minipage}

\begin{minipage}{1\columnwidth}
  \centering
    \subfloat[\textbf{V7}: Action Breaking.]{\includegraphics[width=0.49\columnwidth]{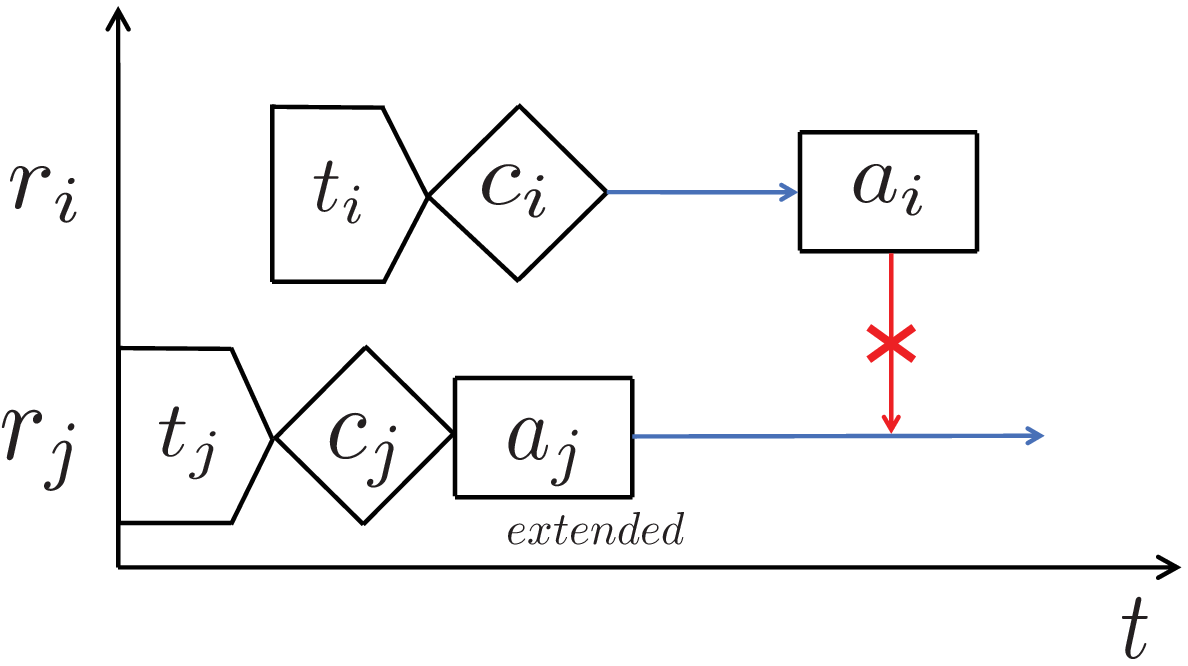}}
    \subfloat[\textbf{V8}: Tardy-channel-based\\ Condition Interference.]{\includegraphics[width=0.49\columnwidth]{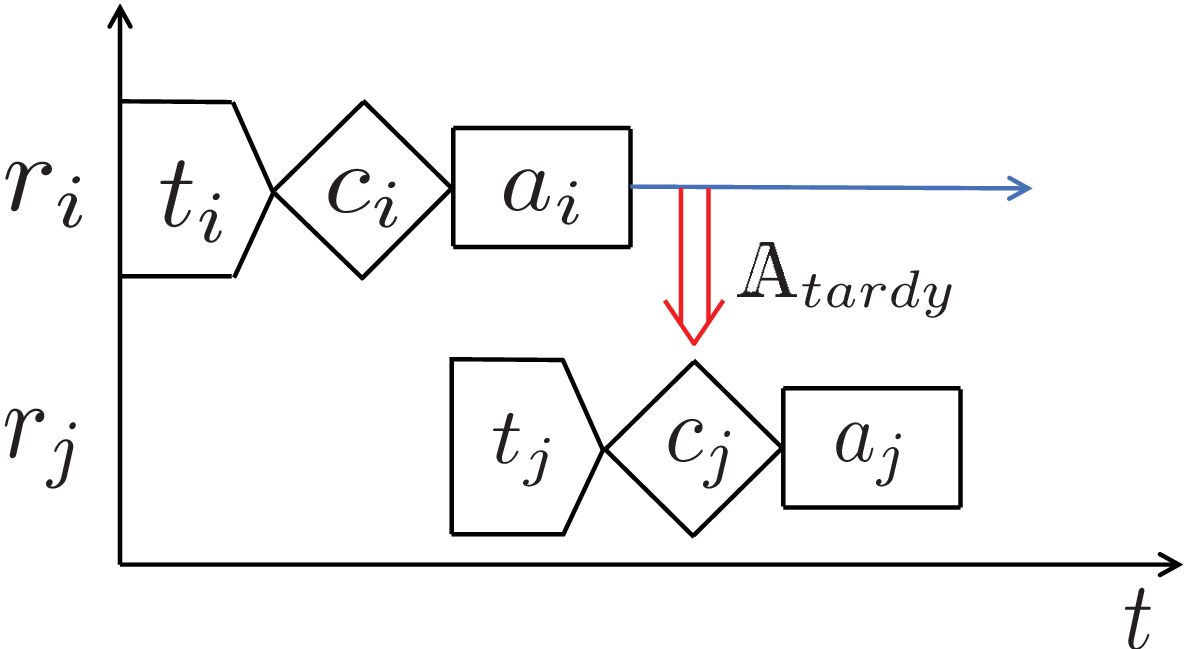}}
\end{minipage}
\vspace{1mm}
\caption{Rule interaction vulnerability patterns.}
\vspace{-2mm}
\label{fig:vul}
\end{center}
\end{figure}

\textbf{V7:} \textit{Action
Breaking.} Different from V3 where both rules are immediate, $r_i$ and $r_j$ contain at least one extended rule ($r_{i}\in r_{ext} \lor r_{j}\in r_{ext}$).
Assume $r_j$ is the extended rule, during its execution, the state of the device in $a_j$ will be stayed for a while. $r_i$ with different latency preference will complete $a_i$ before $r_j$, interrupting the process of extended $a_j$ ($a_{i} \xrightarrow{\times} a_{j}$). \textit{Group4} in Table \ref{tab:frg} shows an example of V7.

\textbf{V8:} \textit{Tardy-channel-based Condition Interference.} Similar to V2, $a_i$ can
change the satisfaction of $c_j$ sharing the same physical channel ($a_{i}\Rightarrow c_{j}$), but from the tardy channel ($a_{i} \cap c_{j} \cap \mathbbm{A}_{tardy} \ne\emptyset$). \textit{Group5} in Table \ref{tab:frg} shows an example of V8.

\section{Model Information Extraction}
\label{asec:mie}

Depending on existing efforts \cite{alhanahnah2020scalable, yu2022tapinspector}, TAPFixer first extracts inter-procedural control flow graphs (ICFGs) from HA apps and then formalizes ICFGs as TAP rules. It utilizes static program analysis for SmartThings apps written in Groovy and NLP techniques for IFTTT applets.

To extract ICFGs from SmartApps, TAPFixer implements an Abstract Syntax Tree (AST) parser upon the Groovy compiler to perform path-sensitive analyses\cite{Dillig_Dillig_Aiken_2008} on AST nodes.For the closed-source SmartThings API during path analysis, TAPFixer reviews the SmartThings API documentation\cite{DeviceCapabilitiesReference} to manually model the APIs in a manner similar to\cite{Corno_DeRussis_MongeRoffarello_2019}. With the AST parser and the handling of closed-source features, ICFGs can be extracted accurately. To convert IFTTT applets into ICFGs, TAPFixer utilizes NLP methods \cite{wang2019charting, Ding_Hu_2018} for string analysis similar to\cite{alhanahnah2020scalable, celik2019iotguard}. It decomposes the applet text into tokens and searches for trigger, condition, and action titles from them. With extracted titles, TAPFixer constructs an ICFG with two nodes, where the input node represents the trigger and condition while the other node represents the action. \revise{Note that conditions in IFTTT can also be specified using the filter code in the form of JavaScript snippets, which requires hard manual effects \cite{celik2019iotguard} and is not publicly available. Hence, we do not analyze the conditions specified in them like \cite{wang2019charting, yu2022tapinspector}. Using runtime techniques \cite{bastys2018if, celik2019iotguard} is a possible solution to improve the capability of TAPFixer in the future.}

To construct TAP rules from ICFGs, TAPFixer first reformulates device capabilities into JSON format. It then provides a crawler, based on Selenium\cite{Selenium}, for users to extract rule configurations not included in source codes. With the obtained information, TAPFixer extracts valid paths from ICFG and converts them into TAP rules similar to \cite{yu2022tapinspector}. It extracts trigger, condition, and action sets from different domains based on the path location in the app (app initialization, app schedule, and event handling). TAPFixer also extracts device areas, rule semantics, and rule configurations from extracted rule information.

\section{Quantification of Physical Channel Interactions}

\revise{Based on physical interactions studied in existing literature \cite{birnbach2019peeves, chi2022delay,chi2020cross,Ding_Hu_2018,ozmen2022discovering,ding2021iotsafe, wang2019charting,yu2022tapinspector},}
we totally model 9 physical channels as follows: temperature, illuminance, motion, smoke, humidity, CO, \ce{CO2}, sound, and weather status. We consider motion and weather status are not affected by actuator executions, while the rest of the physical channels can be affected by actuators. To quantify physical effects, \revise{we collect existing effort analysis results mainly from \cite{birnbach2019peeves, ding2021iotsafe, yu2022tapinspector} to conduct a qualitative analysis in Table \ref{tab:pc} by associating device effects with their affected physical channels. Besides, we also conduct a simple measurement to further improve the qualitative analysis that we deploy devices in our home to collect measurement results and manually identify these effects from measurement results. Considering that the effects will change as the environment changes, we set these quantitative effects to a value range rather than a fixed value and generate several versions of a rule model by randomly selecting values from the range. This can reduce model errors for different environments and improve the availability of detection and repair results. In the future, we can use an online monitor and machine-learning-based (e.g., SVM used in \cite{birnbach2019peeves}) methods to mine effort models of devices and regularly update them at runtime.}

\begin{table}[t!]
\centering
\caption{Configurations of interactions on physical channels.}\label{tab:pc}
\resizebox{\linewidth}{!}
{
\begin{tabular}{|c|c|c|}
\hline
Physical channel                 & Device action       & Physical effect                                                                                    \\ \hline
\multirow{6}{*}{Temperature}     & ACMode.heat         & Rise 1\textcelsius\, in 10-15min                                                                                \\ \cline{2-3}
                                 & ACMode.cool         & Drop 1\textcelsius\, in 10-15min                                                                                \\ \cline{2-3}
                                 & thermostatMode.heat & Rise 1\textcelsius\, in 15-20min                                                                                \\ \cline{2-3}
                                 & thermostatMode.cool & Drop 1\textcelsius\, in 15-20min                                                                                \\ \cline{2-3}
                                 & heater.on           & Rise 1\textcelsius\, in 10-15min                                                                                \\ \cline{2-3}
                                 & window.open         & \begin{tabular}[c]{@{}c@{}}Rise or drop 1\textcelsius\, in 10-15min\\ accroding to Temp difference\end{tabular} \\ \hline
\multirow{4}{*}{Humidity}        & sprinkler.on        & Rise 10\% in 10-15min                                                                              \\ \cline{2-3}
                                 & fan.on              & Drop 10\% in 15-20min                                                                              \\ \cline{2-3}
                                 & humidifier.on       & Rise 10\% in 10-15min                                                                              \\ \cline{2-3}
                                 & dehumidifier.on     & Drop 10\% in 15-20min                                                                              \\ \hline
\multirow{2}{*}{Smoke}           & waterVavle.on       & Clear smoke in 10-15min                                                                                                                                                 \\ \cline{2-3}
                                 & window.open         & Clear smoke in 20-25min                                                                            \\ \hline
\multirow{2}{*}{Carbon monoxide} & fan.on              & Clear CO in 15-20min                                                                               \\ \cline{2-3}
                                 & window.open         & Clear CO in 15-20min                                                                               \\ \hline
\multirow{2}{*}{Carbon dioxide}  & fan.on              & Decrease 1 level in 10-15min                                                                       \\ \cline{2-3}
                                 & window.open         & Decrease 1 level in 10-15min                                                                       \\ \hline
Sound                            & window.close        & Decrease 20 \emph{db}                                                                                       \\  \hline

\multirow{2}{*}{Illuminance}  & light.on              & Increase 100 \emph{lux}                                                                       \\ \cline{2-3}
                                 & light.off         & Decrease 100 \emph{lux}                                                                     \\ \hline
\end{tabular}
}
\end{table}

\section{Template Equivalence of Correctness Properties}
\label{asec:tesp}
Natural language templates of correctness properties can meet the majority of smart home scenarios that users expect. Based on whether properties are conditional and their states and/or events descriptions, there are 9 natural language templates\cite{zhang2019autotap}. These language templates can be summarized as two types of logical templates: \textit{event-based} and \textit{state-based}, shown in Table \ref{tab:lqs}.
The former focuses on identifying and handling exceptions timely, while the latter focuses on continuously preventing exceptions from occurring, which are often combined to ensure safety.
We use four equivalent relations to implement template conversion and prove the soundness of the translation equivalence as follows:

\begin{thm}
    The single-state correctness property is logically equivalent to the multi-state one since it is obviously a special case of multi-state properties where $n$ in $\land_{i=1}^n state_{i}$ is equal to 1.
\end{thm}

\begin{thm}
    An unconditional correctness property is obviously logically equivalent to the corresponding conditional one since it can be viewed as having a condition whose value is $True$.
\end{thm}

\begin{table}[t!]
\centering
\caption{Logically equivalent correctness property types.}
\label{tab:lqs}
\resizebox{\linewidth}{!}
{
\begin{tabular}{|p{2cm}|p{4cm}|p{4cm}|}
\hline
\begin{tabular}[c]{@{}c@{}}Summarised\\ property types\end{tabular} & Property types                      & Natural language templates                                                            \\ \hline
\multirow{3}{*}{Event-based}                                       & One-Event Unconditional            & {[}\textit{event}{]} should {[}\textit{never}{]} happen                                     \\ \cline{2-3}
                                                                   & Event-State Conditional (always)   & {[}\textit{event}{]} should {[}\textit{always}{]} happen when {[}\textit{$state_1$ ,..., $state_n$}{]}       \\ \cline{2-3}
                                                                   & Event-State Conditional (never)    & {[}\textit{event}{]} should {[}\textit{never}{]} happen when {[}\textit{$state_1$ ,..., $state_n$}{]}      \\ \hline
\multirow{6}{*}{State-based}                                       & One-State Unconditional (always)   & {[}\textit{state}{]} should {[}\textit{always}{]} be active                                 \\ \cline{2-3}
                                                                   & One-State Unconditional (never)    & {[}\textit{state}{]} should {[}\textit{never}{]} be active                                  \\ \cline{2-3}
                                                                   & Multi-State Unconditional (always) & {[}\textit{$state_1$, ..., $state_n$}{]} should {[}\textit{always}{]} occur together              \\ \cline{2-3}
                                                                   & Multi-State Unconditional (never)  & {[}\textit{$state_1$, ..., $state_n$}{]} should {[}\textit{never}{]} occur together               \\ \cline{2-3}
                                                                   & State-State Conditional (always)   & {[}\textit{state}{]} should {[}\textit{always}{]} be active while {[}\textit{$state_1$, ..., $state_n$}{]} \\ \cline{2-3}
                                                                   & State-State Conditional (never)    & {[}\textit{state}{]} should {[}\textit{never}{]} be active while {[}\textit{$state_1$, ..., $state_n$}{]}  \\ \hline
\end{tabular}
}
\end{table}

\begin{thm}
    The natural language template (``$[event]$ should $[always]$ happen if $[state_1, ..., state_n]$'') is logically equivalent to ``$[\neg event]$ should $[never]$ happen if $[state_1, ..., state_n]$''.
    \begin{equation}
        G(\land_{i=1}^n state_{i}\Rightarrow X(event)) \equiv \neg F(\land_{i=1}^n state_{i}\land X(\neg event))
\end{equation}
    where $\neg event$ has the constraint opposite to $event$, e.g., $\neg event$ of turning off the heater is turning it on.
\end{thm}
\begin{proof}
The LTL formula of the natural language template describing $[event]$ should $[always]$ happen if $[state_1, ..., state_n]$ is given by
    \begin{equation}
        G(\land_{i=1}^n state_{i}\Rightarrow X(event))
\end{equation}

Applying Law of Excluded Middle $    \psi \Rightarrow \chi \equiv \neg \psi \lor \chi$ to (10), we have that
\begin{equation}
            G(\land_{i=1}^n state_{i}\Rightarrow X(event)) \equiv             G(\neg(\land_{i=1}^n state_{i})\lor X(event))
\end{equation}

Applying De Morgan's Law $    \neg \psi \lor \neg \chi \equiv  \neg (\psi \land \chi)$ to (11), we have that
\begin{equation}
                G(\land_{i=1}^n state_{i}\Rightarrow X(event)) \equiv G(\neg(\land_{i=1}^n state_{i}\land \neg X(event)))
\end{equation}

Applying the negation propagation of X LTL logic $\neg X (\psi) \equiv  X(\neg \psi)$ to (12), we have that
\begin{equation}
                G(\land_{i=1}^n state_{i}\Rightarrow X(event)) \equiv G(\neg(\land_{i=1}^n state_{i}\land X(\neg event)))
\end{equation}

Applying the negation propagation of G LTL logic $G(\neg \psi) \equiv \neg F (\psi)$ to (13), we have that
\begin{equation}
                G(\land_{i=1}^n state_{i}\Rightarrow X(event)) \equiv \neg F(\land_{i=1}^n state_{i}\land X(\neg event))
\end{equation}

Theorem 3 is proved.
\end{proof}

\begin{thm}
    The natural language template (``$[state]$ should $[always]$ be active when $[state_1, ..., state_n]$'') is logically equivalent to ``$[\neg state]$ should $[never]$ be active when $[state_1, ..., state_n]$''.
    \begin{equation}
G(\land_{i=1}^n state_{i}\Rightarrow state) \equiv \neg F(\land_{i=1}^n state_{i}\land \neg state)
\end{equation}
     where $\neg state$ has attribute values except $state$, e.g., $\neg state$ of alarm.siren is off, strobe, or both.
\end{thm}
\begin{proof}
The LTL formula of the natural language template describing $[state]$ should $[always]$ be active when $[state_1, ..., state_n]$ is given by
    \begin{equation}
        G(\land_{i=1}^n state_{i}\Rightarrow state)
\end{equation}

    Applying Law of Excluded Middle $    \psi \Rightarrow \chi \equiv \neg \psi \lor \chi$ to (17), we have that
\begin{equation}
            G(\land_{i=1}^n state_{i}\Rightarrow state) \equiv             G(\neg(\land_{i=1}^n state_{i})\lor state)
\end{equation}

Applying De Morgan's Law $    \neg \psi \lor \neg \chi \equiv  \neg (\psi \land \chi)$ to (16), we have that
\begin{equation}
                G(\land_{i=1}^n state_{i}\Rightarrow state) \equiv G(\neg(\land_{i=1}^n state_{i}\land \neg state))
\end{equation}

Applying the negation propagation of G LTL logic $G(\neg \psi) \equiv \neg F (\psi)$ to (18), we have that
\begin{equation}
                G(\land_{i=1}^n state_{i}\Rightarrow state) \equiv \neg F(\land_{i=1}^n state_{i}\land \neg state)
\end{equation}

Theorem 4 is proved.
\end{proof}

\section{Categorized and Prioritized Correctness Properties}

Scenario-based correctness property categorization ensures the \revise{safety} of different devices operating in the same automation scenario as described in Section \ref{marketapps}. To address property conflicts, TAPFixer develops prioritized correctness properties. It first sorts according to pre-proposition priority and then according to post-proposition priority. \revise{TAPFixer defines the pre-proposition priority based on automation scenarios as shown in Table  \ref{tab:sortingdescriptions1} and post-proposition priority based on device capabilities as shown in Table  \ref{tab:sortingdescriptions2}. It requires one-pass manual efforts and many of them are reusable across different scenarios since home device types and usage scenarios are limited.}
 Prioritized properties that share the same device capability are listed in ascending priority order in Table \ref{tab:sp}.
The grouping of scenario-based \emph{G1-G7} and prioritized \emph{G8-G23} are also shown in Table \ref{tab:sp}.

TAPFixer defines state-based properties to continuously prevent exceptions from occurring (e.g., P.10, P.34, P.44, P.53). The safety-sensitive properties do not specify a latency (e.g., P.10, P.29, P.31, P.42) because it is expected that these safety measures will remain effective until the risk is eliminated. For instance, the security camera in P.10 is expected to work at all times while the user is away, rather than only for a limited period. Whereas other properties can be designed to be satisfied periodically (e.g., P.34, P.44), allowing for permitted latencies to be specified. If there is a tardy attribute in the pre-position of a property, it takes a period for the satisfiability of the pre-position to change from true to false (e.g., \ce{CO2} drops below the defined threshold in P.34). We define that the latency longer than it is not permitted.
 Properties with latencies follow state-based property in TAPFixer for violation detection.
Take Fig. \ref{fig:rin}(a) and P.34 with a latency (“fan should be on for at least 10min”) for example, TAPFixer defines a variable ``fan.timer'' to record the remaining operating time of the fan (see $\S$\ref{subsec:modeling}) and translates P.34 into the LTL form: G(\ce{CO2} > a predefined value $\land$ (fan.timer $\geq$ fan.config\_latency-600) $\Rightarrow$fan.on), which follows the LTL template of the state-based property in Table \ref{tab:tspnp}.

\begin{table}[]
\caption{\revise{Sorting descriptions of the pre-proposition priority.}}
\label{tab:sortingdescriptions1}
\resizebox{\linewidth}{!}{
\begin{tabular}{|c|c|}
\hline
\begin{tabular}[c]{@{}c@{}}Scenarios in the pre-proposition \\ of the correctness property\end{tabular} &
  Pre-proposition priority \\ \hline
General &
  \begin{tabular}[c]{@{}c@{}}user.not\_present \textgreater user.present,\\ smoke.detected = CO.detected \textgreater weather.raining \textgreater\\ \ce{CO2}-related = humidity-related\end{tabular} \\ \hline
Temperature-related &
  \begin{tabular}[c]{@{}c@{}}user.not\_present \textgreater heater.on = AC.on \textgreater the temperature \\ is below / rises above a predefined value\end{tabular} \\ \hline
\end{tabular}}
\end{table}

\begin{table}[]
    \centering
    \caption{\revise{Sorting descriptions of the post-proposition priority.}}
    \label{tab:sortingdescriptions2}
    \resizebox{\linewidth}{!}{
    \begin{tabular}{|c|c|}
    \hline
        Device Capabilities & Post-proposition Priority \\ \hline
        light.switch & light.on = light.off \\ \hline
        door.lock & door.lock > door.unlock \\ \hline
        security\_camera.switch & camera.on > camera.off \\ \hline
        switch.switch & switch.off > switch.on \\ \hline
        AC.mode & AC.heating\_mode = AC.cooling\_mode \\ \hline
        heater.switch & heater.off > heater.on \\ \hline
        coffee\_machine.switch & coffee\_machine.off > coffee\_machine.on \\ \hline
        electric\_blanket.switch & electric\_blanket.off > electric\_blanket.on \\ \hline
        alarm.state & alarm.activated > alarm.unactivated \\ \hline
        ventilating\_fan.switch & ventilating\_fan.off > ventilating\_fan.on \\ \hline
        oven.switch & oven.off > oven.on \\ \hline
        gas\_water\_heater.switch & gas\_water\_heater.off > gas\_water\_heater.on \\ \hline
        gas\_valve.switch & gas\_valve.off > gas\_valve.on \\ \hline
        water\_valve.switch & water\_valve.on > gas\_valve.off \\ \hline
        sprinkler.switch & sprinkler.on = sprinkler.off \\ \hline
        window.switch & window.open = window.close \\ \hline
    \end{tabular}}
\end{table}

\renewcommand{\arraystretch}{1.0}
\begin{table}[]
\caption{Categorized and prioritized correctness properties.} 
\label{tab:sp}
\centering
\resizebox{\columnwidth}{!}
{
\begin{tabular}{|c|p{8cm}|C{2cm}|C{2cm}|}
\hline
\multirow{3}{*}{Property} & \multirow{3}{*}{Description}  &  Scenario-based category types & Priority-based category types\\ \hline
P.1 &   IF the user arrives home, the light should be turned on. & G1 &G8\\ \hline
P.2 &   IF the user is not at home / not nearby-home, the light should be turned off.& G2&G8\\ \hline
P.3 &   WHEN the user is not at home / not nearby-home, the light should be off.&G2 &G8\\ \hline
P.4 &   IF the user arrives home, the garage door should be opened.& G1 & G9\\ \hline
P.5 &   IF the user leaves home, the garage door should be closed.& G2& G9\\ \hline
P.6 &   WHEN the user leaves home, the garage door should be closed.& G2& G9\\ \hline
P.7 &   IF the user is not at home / not nearby-home, the door should be locked.& G2& G9\\ \hline
P.8 &   WHEN the user is not at home / not nearby-home, the door should be locked.& G2& G9\\ \hline
P.9 &   IF the user is not at home / not nearby-home, the security camera should be turned on.& G2&G10\\ \hline
P.10 &   WHEN the user is not at home / not nearby-home, the security camera should be on.& G2&G10\\ \hline
P.11 &   IF the door opens while the user is not at home / not nearby-home, the security camera should take pictures.&G2 &G10\\ \hline
P.12 &   IF the user is not at home / not nearby-home, the switch should be turned off.&G2 &G11\\ \hline
P.13 &   WHEN the user is not at home / not nearby-home, the switch should be off.& G2&G11\\ \hline
P.14 &   IF the temperature is below a predefined value and someone is at home, the AC should be in heating mode.& G3&G12\\ \hline
P.15 &   IF the temperature rises above a predefined value, the AC should be in cooling mode.&G3 &G12\\ \hline
P.16 &   WHEN the heater is on, the AC should be off.& G3 &G12\\ \hline
P.17 &  IF the user is not at home / not nearby-home, the AC should be turned off.& G2&G12\\ \hline
P.18 &   WHEN the user is not at home / not nearby-home, the AC should be off.& G2&G12\\ \hline
P.19 &   IF the temperature is below a predefined value while someone is at home, the heater should be turned on.&G3 &G13\\ \hline
P.20 &   IF the temperature rises above a predefined value, the heater should be turned off.&G3 &G13\\ \hline
P.21 &   WHEN the AC is on, the heater should be off.& G3&G13\\ \hline
P.22 &   IF the user is not at home / not nearby-home, the heater should be turned off.&G2 &G13\\ \hline
P.23 &   WHEN the user is not at home / not nearby-home, the heater should be off.&G2 &G13\\ \hline
P.24 &   IF the user is not at home / not nearby-home, the coffee machine should be turned off.&G2 &G14 \\ \hline
P.25 &   WHEN the user is not at home / not nearby-home, the coffee machine should be off.&G2 & G14\\ \hline
P.26 &   IF the user is not at home / not nearby-home, the electric blanket should be turned off.&G2 & G15\\ \hline
P.27 &   WHEN the user is not at home / not nearby-home, the electric blanket should be off. &G2 &G15\\ \hline
P.28 &   IF the smoke is detected, the alarm should be activated.& G4&G16\\ \hline
P.29 &   WHEN there is smoke, the alarm should be activated.& G4&G16\\ \hline
P.30 &   IF CO is detected, the alarm should be activated.&G5 &G16\\ \hline
P.31 &   WHEN CO is detected, the alarm should be activated. &G5 &G16\\ \hline
P.32 &   IF humidity is greater than a predefined value, the ventilating fan should be turned on.& G6 &G17\\ \hline
P.33 &   IF \ce{CO2} is greater than a predefined value, the ventilating fan should be turned on.& G4&G17\\ \hline
P.34 &   WHEN \ce{CO2} remains greater than a predefined value, the ventilating fan should be on for at least the permitted time. &G4 &G17\\ \hline
P.35 &   IF the user is not at home / not nearby-home, the oven should be turned off. & G2&G18\\ \hline
P.36 &   WHEN the user is not at home / not nearby-home, the oven should be off. & G2&G18\\ \hline
P.37 &   IF CO is detected, the natural gas hot water heater should be turned off. &G5 & G19\\ \hline
P.38 &   WHEN CO is detected, the natural gas hot water heater should be off. &G5 &G19 \\ \hline
P.39 &   IF CO is detected, the gas valve should shut off.&G5 & G20\\ \hline
P.40 &   WHEN CO is detected, the gas valve should be closed.&G5 & G20 \\ \hline
P.41 &   IF the smoke is detected, the water valve should be turned on.&G4 &G20\\ \hline
P.42 &   WHEN there is smoke, the water valve should be on.& G4&G20\\ \hline
P.43 &   IF the soil moisture sensor is below a predefined value, the sprinkler system should be turned on.& G6&G21\\ \hline
P.44 &   When the soil moisture sensor is below a predefined value, the sprinkler system should be on for at least the permitted time.& G6&G21\\ \hline
P.45 &   IF the weather is raining, the sprinkler should be turned off.& G7&G21\\ \hline
P.46 &   WHEN the weather is raining, the sprinkler should be off.&G7 &G21\\ \hline
P.47 &   IF \ce{CO2} is greater than a predefined value, the window should be opened.& G5&G22\\ \hline
P.48 &   IF the weather is raining, the window should be closed.&G7 &G22\\ \hline
P.49 &   WHEN the weather is raining, the window should be closed.&G7 &G22\\ \hline
P.50 &   IF the smoke is detected, the window should be opened.& G4&G23\\ \hline
P.51 &   WHEN there is smoke, the window should be opened.& G4&G23\\ \hline
P.52 &   IF CO is detected, the window should be opened. &G5 &G23\\ \hline
P.53 &   WHEN CO is detected, the window should be opened. &G5 &G23\\ \hline
\end{tabular}}
\end{table}

\end{document}